\documentclass[a4paper]{article}
\usepackage{amsmath, amssymb, amsthm, bm, hyperref, graphicx, listings, booktabs, xcolor, cleveref, thm-restate, subcaption}

\usepackage[margin=1in]{geometry}

\definecolor{ZP}{HTML}{FF8C00}

\title{Modeling Loss-Versus-Rebalancing in Automated Market Makers via Continuous-Installment Options}

\newtheorem{theorem}{Theorem}
\newtheorem{lemma}[theorem]{Lemma}

\newcommand{\LVR}{\mathrm{LVR}}
\newcommand{\Fee}{\mathrm{Fee}}
\newcommand{\dd}{\,\mathrm d}
\newcommand{\EE}{\mathbb E}
\newcommand{\QQ}{\mathbb Q}
\newcommand{\RR}{\mathbb R}

\DeclareMathOperator{\Var}{Var}
\newcommand{\bb}[1]{\left(#1\right)}
\newcommand{\bs}[1]{\left\{#1\right\}}
\newcommand{\calF}{\mathcal{F}}

\author{
Srisht Fateh Singh\\
University of Toronto, Canada\\
\href{mailto:srishtfateh.singh@mail.utoronto.ca}{\tt srishtfateh.singh@mail.utoronto.ca}
\and
Reina Ke Xin Li\\
University of Toronto, Canada\\
\href{mailto:reinakx.li@mail.utoronto.ca}{\tt reinakx.li@mail.utoronto.ca}
\and
Samuel Gaskin\\
University of Toronto, Canada\\
\href{mailto:samuel.gaskin@nbc.ca}{\tt samuel.gaskin@nbc.ca}
\and
Yuntao Wu\\
University of Toronto, Canada\\
\href{mailto:winstonyt.wu@mail.utoronto.ca}{\tt winstonyt.wu@mail.utoronto.ca}
\and
Jeffrey Klinck\\
University of Toronto, Canada\\
\href{mailto:jeffrey.klinck@mail.utoronto.ca}{\tt jeffrey.klinck@mail.utoronto.ca}
\and
Panagiotis Michalopoulos\\
University of Toronto, Canada\\
\href{mailto:p.michalopoulos@mail.utoronto.ca}{\tt p.michalopoulos@mail.utoronto.ca}
\and
Zissis Poulos\\
York University, Canada\\
\href{mailto:zpoulos@yorku.ca}{\tt zpoulos@yorku.ca}
\and
Andreas Veneris\\
University of Toronto, Canada\\
\href{mailto:veneris@eecg.toronto.edu}{\tt veneris@eecg.toronto.edu}
}

\begin{document}

\maketitle

\begin{abstract}
This paper mathematically models a constant-function automated market maker (CFAMM) position as a portfolio of exotic options, known as perpetual American continuous-installment (CI) options. This model replicates an AMM position's delta at each point in time over an infinite time horizon, thus taking into account the perpetual nature and optionality to withdraw of liquidity provision. This framework yields two key theoretical results: (a) It proves that the AMM's adverse-selection cost, loss-versus-rebalancing (LVR), is analytically identical to the continuous funding fees (the time value decay or theta) earned by the at-the-money CI option embedded in the replicating portfolio. (b) A special case of this model derives an AMM liquidity position's delta profile and boundaries that suffer approximately constant LVR, up to a bounded residual error, over an arbitrarily long forward window. Finally, the paper describes how the constant volatility parameter required by the perpetual option can be calibrated from the term structure of implied volatilities and estimates the errors for both implied volatility calibration and LVR residual error. Thus, this work provides a practical framework enabling liquidity providers to choose an AMM liquidity profile and price boundaries for an arbitrarily long, forward-looking time window where they can expect an approximately constant, price-independent LVR. The results establish a rigorous option-theoretic interpretation of AMMs and their LVR, and provide actionable guidance for liquidity providers in estimating future adverse-selection costs and optimizing position parameters.
\end{abstract}

\section{Introduction}\label{sec:intro}
The success of blockchains supporting smart contract such as Ethereum~\cite{buterin2013ethereum}, Solana~\cite{yakovenko2018solana}, etc., has led to the rise of {\em Decentralized Finance} (DeFi) which offers alternatives to traditional financial services by removing central trusted intermediaries and replacing them with public, verifiable, and immutable computer programs. One of the pivotal components of the DeFi infrastructure stack is automated market makers, or AMMs, allowing the exchange of one token for another at prices decided by an underlying algorithm. In recent years, AMMs have seen a rapid adoption reflected in financial metrics such as total value locked (above $\$21B$) and yearly transaction volumes (above $\$2T$), as well as by their composability~\cite{defillama_dexs,dune_dex_metrics}. Today, tens of thousands of tokens are listed and hundreds of applications are built on top of them~\cite{numberOfUniPools}.

The key participants in an AMM are traders and agents known as \emph{liquidity providers}, or LPs. Traders exchange one token for another, where the token pair generally consists of a risky asset with volatile value, and a stable asset or numéraire.\footnote{Although some token pairs consist of two stable assets, this work focuses primarily on pairs of one risky and one stable asset.} On the other hand, LPs serve as counterparties to traders (sellers to buyers and buyers to sellers) by depositing both tokens upfront to the exchange. As a result of each trade, LPs receive the less favourable of the two tokens. To hedge against the adverse selection faced by AMMs, LPs can continuously rebalance an off-chain replicating portfolio by accumulating the risky token as its price rises and selling it when the price falls. However, this hedge is not perfect when the AMM is not the primary venue for price discovery because the pool’s quoted price tends to lag behind the prevailing price on a primary venue such as a centralized exchange.

Under such a setting, even proactive LPs who rebalance in response to price changes are exposed to a systematic cost. Since AMM quotes lag those on primary markets, arbitrageurs can act faster than LPs and restore the pool price to the external market level. This generates a small but persistent transfer of value from LPs to arbitrageurs, and when aggregated over multiple price updates, this cost becomes significant. This cost is referred to as \emph{loss-versus-rebalancing} (LVR)~\cite{milionis2022automated}. The rate at which LVR accumulates depends on the steepness of the AMM curve and the volatility of the underlying token; both amplify the arbitrage gap and thus accelerate LVR.

Despite such costs, LPs are incentivized to participate through the earning of trading fees that are proportional to the value of each trade and paid by the trader. Therefore, LPs considering whether to provide liquidity on an AMM pool must calculate their expected \emph{payoff} by estimating and comparing their position's LVR with the anticipated trading fees in some predetermined forward time window. Recent work~\cite{milionis2022automated} have analyzed and quantified expressions for instantaneous LVR and retrospectively tested with historic market data. However, there is limited work that provides estimation methods for future LVR. 

This work estimates the LVR for an LP that decides to provide liquidity for an arbitrarily long period and can exit at any point. It does so by mathematically modeling liquidity provision on a general class of AMMs, known as constant function AMMs (CFAMMs), as selling a continuum of perpetual American put options across continuous strikes. Perpetual American options are financial derivatives that give their holder the right to buy (known as a call) or sell (known as a put) an underlying asset at an agreed-upon price (strike) with no expiration. Unlike traditional vanilla options, this work uses exotic options in its model, known as continuous installment (CI) options, in which the holder must pay a stream of constant installment rate, referred to as \emph{funding fees}, to keep their position alive. This funding fee is analogous to the time value decay of traditional fixed-term options. The paper uses perpetual American CI options because, unlike fixed-term vanilla options, the pricing function of these options does not change over time (assuming other market parameters are constant). Moreover, despite their exotic nature, CI options have been well-studied in the past, and this work builds on results from the existing literature~\cite{ciurlia2009note}. This approach yields two key theoretical results.

\noindent \textbf{Funding Fees = LVR:} In the limit where the installment rate tends to infinity (analogous to extremely short-dated fixed-term options), a CI put option has the classic hockey-stick payoff function: it pays the difference between the strike and the spot price when the underlying's spot price is below the strike price, and zero otherwise. Therefore, in this limit, a continuous distribution of CI puts \emph{exists} that delta-replicates an arbitrary LP position's payoff \emph{at each point in time}. Moreover, the installment rate earned on this distribution of puts (which, at a given point in time, is earned by the option whose strike equals the spot price at that time) reproduces the expression for the instantaneous LVR. Therefore, theoretically, an LP can hold such an options portfolio to stay delta-neutral---the cost due to time value decay of holding this portfolio, \emph{i.e.}, funding fees, is the LVR.

Since the above model is theoretical, as a continuum of options cannot be reproduced in practice, the paper subsequently quantifies the approximation error when a discrete portfolio of options with discrete strike prices is used to replicate the payoff of a liquidity position. 

\noindent\textbf{Constant Future LVR:} The second result analyzes the converse scenario where a liquidity position's payoff replicates the valuation of a single perpetual American CI option. This produces a unique liquidity profile with \emph{almost constant} instantaneous LVR rate over a forward time window. Moreover, this rate is approximately equal to the funding fee of the replicated CI option. As a result, this yields guidance to LPs, under the model assumptions, on:
\begin{itemize}
    \item Choosing the price boundaries and shape of liquidity provision that incurs \emph{predictable, flat, price-path independent} future LVR.
    \item Estimating the forward adverse-selection cost for a planned holding period.
    \item Understanding the relationship between the optimal holding period and the width of the liquidity position.
    \item Selecting an appropriate pool based on its expected future trading fee income.
\end{itemize}

The paper is organised as follows. Section~\ref{sec:background} explains notations and the necessary background, Section~\ref{sec:related-work} discusses prior literature and related works, Section~\ref{sec:motivation} motivates the option-based interpretation, Section~\ref{sec:kernel-replication} provides the options decomposition, Section~\ref{sec:LVR-from-CI} proves the funding-fee--LVR identities, Section~\ref{sec:error-analysis-delta} measures the approximation error on delta when the continuous strip is replaced by finitely many strikes, and Section~\ref{sec:sigma-calibration-put} presents the volatility calibration and design rules for LPs. Finally, Section~\ref{sec:conclusion} concludes with directions for future research.

\section{Background}\label{sec:background}
In this section, we provide the necessary notation, terminology and background concepts used in the remainder of the paper.
\subsection{Notation}\label{sec:notation}
We consider two tokens: $token\;0$ representing a risky asset (\emph{e.g.}~BTC, ETH) and $token\; 1$ representing a stable/safe asset (\emph{e.g.}~USDC). Let $S_t$ denote the spot price of $token\;0$ in units of $token\; 1$ at time $t$ that
follows a geometric Brownian motion (GBM) on a filtered probability space $\bb{\Omega, \calF, \bs{\calF_t}_{t\geq 0}, \QQ}$ satisfying the standard assumptions for GBMs (where $\QQ$ is a risk-neutral probability measure), so that
\begin{equation}\label{eq:gbm}
  \frac{\dd S_t}{S_t}=r\,\dd t+\sigma\,\dd B_t^{\QQ},
\end{equation}
with constant annual risk-free rate~$r$ and volatility~$\sigma>0$. $\{B^{\mathbb{Q}}_t\}_{t\ge0}$ is a Wiener process. As usual, we assume that AMMs constitute secondary markets and the price $S_t$ is governed by primary markets such as centralized exchanges. In the subsequent sections, we omit the subscript $t$ and use $S$ and $S_t$ interchangeably for convenience. 

\subsection{Constant-Function Automated Market Makers (CFAMMs)}\label{sec:cfamm}
A \emph{constant-function automated market maker} (CFAMM) maintains token reserves $(x,y)$, deposited by LPs, such that each trade transforms the reserves to $(x', y')$ and the reserves before and after the trade satisfy an invariant function $F(x,y)=F(x',y') = k$. In a \emph{constant-product} AMM (CPAMM), such as Uniswap v2~\cite{uniswapv2}, the invariant takes the form of $F(x,y)=\sqrt{xy}$. As a result, the marginal exchange price, assuming no arbitrage, takes the form: $S = -\frac{\partial y}{\partial x} = \frac{y}{x}$. This is shown in Figure~\ref{fig:cpamm-invariant}. Therefore, the expression for token reserves are $x = \frac{k}{\sqrt{S}}$, and $y = k\sqrt{S}$.

A CPAMM with concentrated liquidity (as in Uniswap~v3~\cite{uniswapv3}) uses an invariant parameter~$k$ within a price band $[a, b]$. The token reserves of LPs in this band consist of only \emph{token~0} when the price $S \geq b$, only \emph{token~1} when $S \leq a$, and for prices $S \in (a, b)$, the reserves are:
\begin{align}
    x = k\frac{\sqrt{b} - \sqrt{S}}{\sqrt{Sb}},\qquad y = k\left(\sqrt{S} - \sqrt{a}\right).
\end{align}
Therefore, the reserve value (denominated in token~1) is
\begin{equation}\label{eq:V}
  V(S) = k\left(\sqrt{S} - \sqrt{a}\right) + kS\,\frac{\sqrt{b} - \sqrt{S}}{\sqrt{Sb}}, \qquad S \in [a, b].
\end{equation}
The liquidity position's sensitivity to price, known as it's \emph{delta}, is denoted by $X(S)$, and the sensitivity of delta to price, known as \emph{gamma}, is denoted by $\Gamma(S)$. In practice, delta is a measure of exposure to small changes in the price of the risky asset, whereas gamma is a measure of exposure to large movements of the risky asset's price. These are given by the first and second derivatives of the value function $V(S)$ with respect to the underlying's price, respectively, and are expressed as follows:
\begin{equation}\label{eq:delta-gamma}
  X(S) := V'(S),\qquad \Gamma(S) := V''(S) = X'(S) \le 0 \quad (S \in (a, b)).
\end{equation}
Figure~\ref{fig:delta-gamma} illustrates the behavior of delta and gamma across the liquidity band. As shown, the magnitudes of both delta and gamma decrease monotonically with price, as higher prices correspond to a greater allocation to the numeraire asset.

\begin{figure}[t]
  \centering
  \begin{subfigure}[b]{0.5\textwidth}
    \centering
    \includegraphics[width=\textwidth]{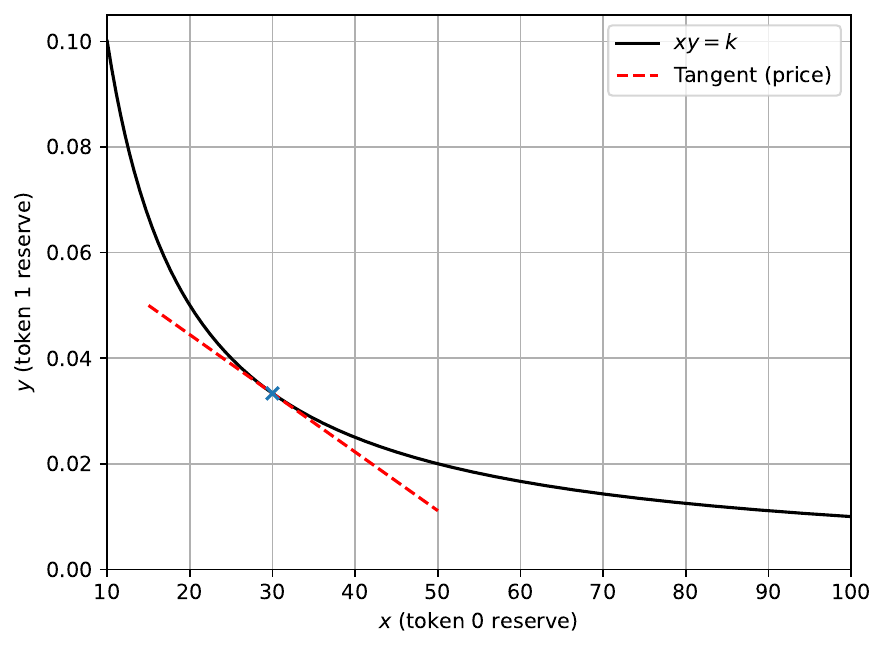}
    \caption{CPAMM Invariant with $k=1$ and Marginal Price (Tangent at $x=30$)}
    \label{fig:cpamm-invariant}
  \end{subfigure}
  \hfill
  \begin{subfigure}[b]{0.46\textwidth}
    \centering
    \includegraphics[width=\textwidth]{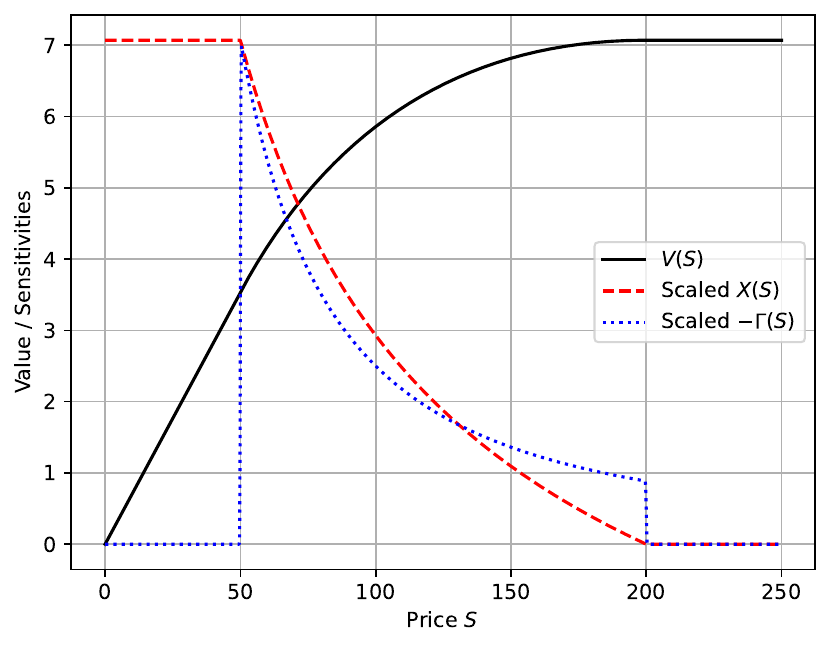}
    \caption{$V(S)$, $100X(S)$, and $-5000\Gamma(S)$ for concentrated liquidity position in price band $[50,200]$, and $k = 1$.}
    \label{fig:delta-gamma}
  \end{subfigure}
  \caption{Profiling CPAMM invariant and concentrated liquidity position, its delta and gamma.}
  \label{fig:CPAMM-profile}
\end{figure}

\subsubsection{Loss-Versus-Rebalancing (LVR)}
A continuously rebalanced, self-financing delta-hedge portfolio that holds $X(S_t)$ units of token~0 has value $W_t$ with $\dd W_t=X(S_t)\dd S_t$.  The difference
\begin{equation}
  \mathrm{LVR}_t:=V(S_t)-W_t
\end{equation}
quantifies the AMM's adverse-selection loss relative to a hedged trader and is referred to as loss-versus-rebalancing or LVR. The instantaneous LVR, $\dd\mathrm{LVR}_t$, grows quadratically with the spot price and volatility, and linearly with the gamma of the liquidity position~\cite{milionis2022automated}. 
\begin{equation}\label{eq:dLVR}
  \dd\mathrm{LVR}_t=\tfrac12\,\sigma^{2}S_t^{2}\Gamma(S_t)\,\dd t=\tfrac12\,\sigma^{2}S_t^{2}\bigl[X'(S_t)\bigr]\dd t.
\end{equation}

We will later demonstrate the equivalence between the right-hand side of Eq.~\eqref{eq:dLVR} and the funding fee of a perpetual American CI option.
\subsection{Perpetual American Continuous-Installment Options}\label{sec:ci-options}

A \emph{perpetual American continuous-installment put option} has no expiration date and requires the holder to pay a continuous stream of \emph{constant funding fee}~$q > rK$ per year to keep the contract alive. At any point, the holder may choose to stop paying the fee, at which point they can either exercise the option or drop the position. The option can be exercised at any time for a payoff of $\max(K - S, 0)$~\cite{ciurlia2009note}. In the analysis below, we assume the underlying asset (token~0) pays zero dividends.

\subsubsection{Notation and Ordinary Differential Equation Formulation}\label{sec:ode-options}
Let $P_q(S;K)$ denote the discounted put option value at spot price $S$, strike $K$, and funding fee $q$. Let $S_\ell$ denote the lower boundary, below which the option value equals its payoff, and let $S_u$ denote the upper boundary, above which the option value is zero, as illustrated in Figure~\ref{fig:ci-curve}.
Under the risk–neutral dynamics, $P_q(S;K)$ satisfies the inhomogeneous Black-Scholes ordinary differential equation in the \emph{continuation region} $S_\ell< S < S_u$:
\begin{equation}\label{eq:CI-ODE}
  \frac{1}{2}\,\sigma^{2}S^{2}\,
     \frac{\partial^{2}P_q}{\partial S^{2}}
  + rS\,
     \frac{\partial P_q}{\partial S}
  - r\,P_q
  \;=\; q,
  \qquad S\in(S_\ell,S_u).
\end{equation}
The left and right boundaries are determined endogenously from the value–matching and delta-matching conditions
\begin{equation}\label{eq:boundaries-FBP}
  \begin{aligned}
    P_q(S_\ell;K) &= K - S_\ell,  
    &\qquad \frac{\partial P_q}{\partial S}(S_\ell;K) &= -1, \\[4pt]
    P_q(S_u;K)    &= 0,
    &\qquad \frac{\partial P_q}{\partial S}(S_u;K)    &= 0.
  \end{aligned}
\end{equation}
Solving \eqref{eq:CI-ODE} with the four boundary conditions in
\eqref{eq:boundaries-FBP} yields the closed-form expressions below.

\begin{figure}[t]
  \centering
  \begin{subfigure}[b]{0.48\textwidth}
    \centering
    \includegraphics[width=\textwidth]{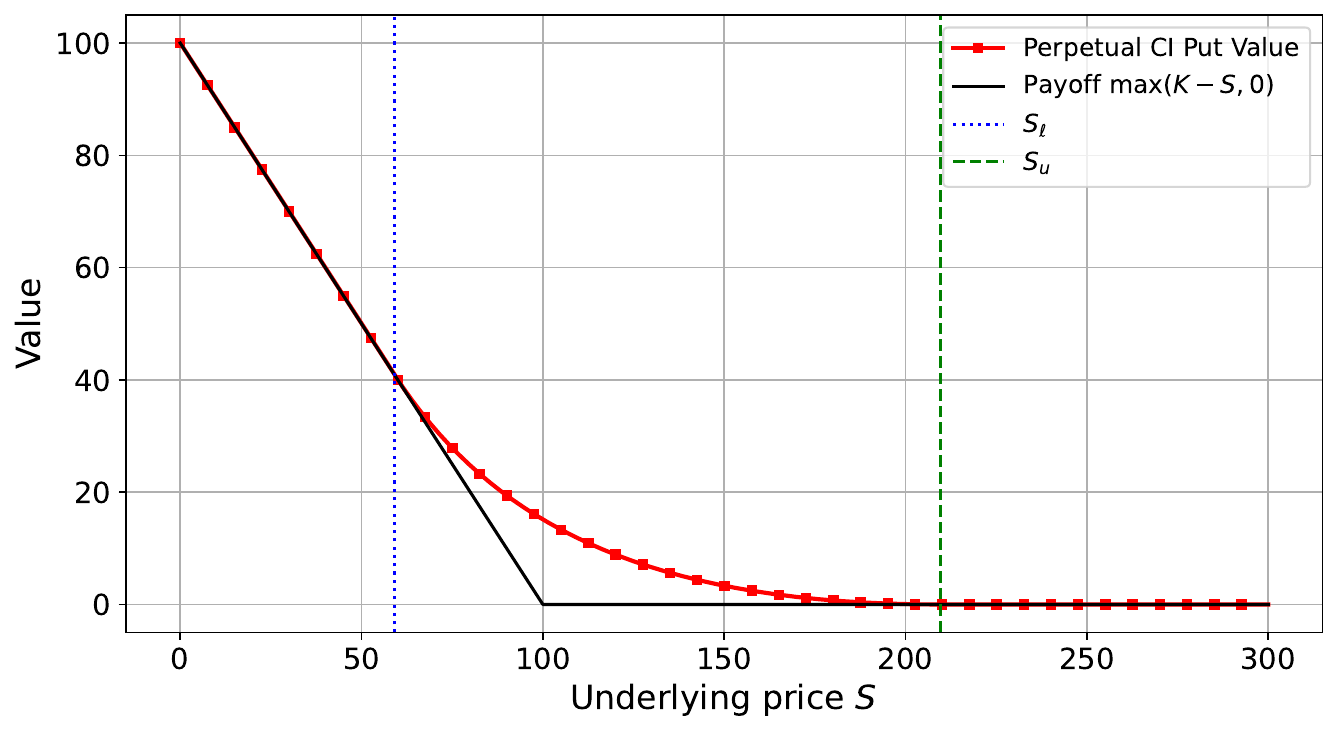}
    \caption{Discounted value of a perpetual CI put option with lower and upper boundaries.}
    \label{fig:ci-curve}
  \end{subfigure}
  \hfill
  \begin{subfigure}[b]{0.48\textwidth}
    \centering
    \includegraphics[width=\textwidth]{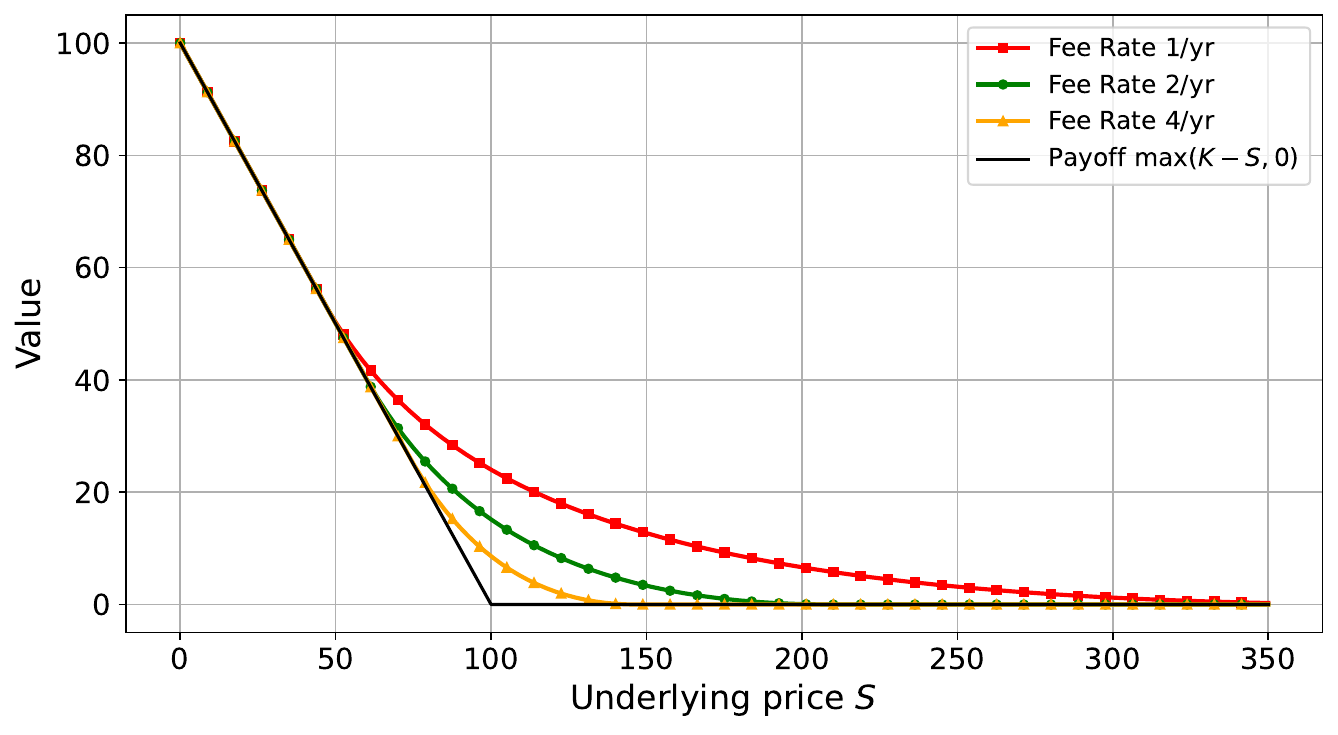}
    \caption{Discounted value of a perpetual CI put option with fee rates $1$, $2$, $4$ per annum.}
    \label{fig:ci-curve-q}
  \end{subfigure}
  \caption{Profiling of perpetual CI put option with $K = 100.0$, 
$\sigma = 0.25$, $q=2$ (in the first figure) and $r = 0.01$ (all parameters are annualized).}
  \label{fig:CI-option-profile}
\end{figure}
\subsubsection{Closed-form solution}
The expression for the option price $P_{q}(S;K)$ takes the following closed form as derived in~\cite{ciurlia2009note}.
\begin{equation}\label{eq:Pq}
    P_q(S;K)=\alpha_pS+\beta_pS^{\gamma_p}+\frac{q}{r}
\end{equation}
The delta of the put option, $X_q(S;K)= \frac{\partial}{\partial S} P_q(S;K)$, thus takes the form:
\begin{equation}\label{eq:DeltaCIput}
  X_q(S;K)=\alpha_p+\beta_p\gamma_pS^{\gamma_p-1}.
\end{equation}
Moreover, the upper and lower boundaries, $S_{u}$ and $S_{\ell}$ respectively, have the following closed form:
\begin{align}
  S_\ell &= \frac{q}{r+\sigma^{2}/2} \bigl[g - g^{1/\gamma_p}\bigr], \label{eq:boundary_lower}\\
  S_u    &= \frac{q}{r+\sigma^{2}/2} \bigl[g^{1-1/\gamma_p} - 1\bigr].\label{eq:boundary_upper}
\end{align}
Here, $\alpha_p$, $\beta_p$, $\gamma_p$, and $g$ are expressions that depend on parameters $r$, $\sigma$, $K$, and $q$ and their expressions are provided in Appendix~\eqref{app:closed-form}.
 
The Black–Scholes partial differential equation in~\eqref{eq:CI-ODE} contains no derivative with respect to time. Consequently, the value of a perpetual CI option is \emph{time-invariant}. This makes it unique from vanilla finite expiry American or European options, whose pricing has a time-varying component. The analog of expiration in CI options is the funding rate, where  high funding rates make it ``behave similar'' to short expiration option and vice versa for small fee rates. This is depicted in Figure~\ref{fig:ci-curve-q}, where increasing the fee rate reduces the price of the CI option closer to its payoff. The lower and upper boundaries, \( S_\ell \) and \( S_u \), characterize the holder’s optimal policy. When the spot price first falls below \( S_\ell \), it is optimal to \emph{exercise} the option; when it first exceeds \( S_u \), it is optimal to \emph{drop} the option---\emph{i.e.}, to exit the position with zero payoff. This is because, in both cases, the expected benefit of continued funding is outweighed by its cost. In either scenario, the holder stops paying the funding fee immediately upon exit. For this reason, \(S_\ell\) and \(S_u\) are also called \emph{optimal exercise} and \emph{dropping boundaries}, respectively. Conversely, the option seller receives the continuous funding fee only while the spot price remains in the continuation region \(S_\ell < S < S_u\).  We will exploit this fact in the sections that follow.

\section{Related Work}\label{sec:related-work}
Early work on studying LP positions in CFAMMs focus on mitigating risks associated with impermanent loss---the loss experienced by a liquidity provider compared to simply holding the asset, so that an LP can earn trading fees without exposure to this loss. Deng et al.~\cite{deng2023static}  and Fukasawa et al.~\cite{fukasawa2023weighted} study static replication strategies for impermanent loss in finite-time CFAMM positions using European options and variance swaps, respectively. Lipton et al.~\cite{lipton2023unified} further studies model-based dynamic replication. However, it can be argued that a comparison against a buy-and-hold strategy is insufficient as it does not account for adverse selection, where informed traders extract value from passive LPs over time.

Another approach to quantifying LP loss is the concept of LVR, which captures the loss incurred by a liquidity provider compared to continuously rebalancing their portfolio at market prices, due to adverse selection. LVR was formalized by Millionis et al.~\cite{milionis2022automated}. Their work offers closed-form expressions for instantaneous LVR in a CFAMM and provides empirical validation using historical market data. However, they focus on instantaneous and historical LVR and do not address forward-looking or long-term LVR estimation. Meanwhile, our approach provides a framework that captures forward-looking LVR by delta-replication with options portfolios.

Maire and Wunsch~\cite{maire2022market} make a case for the hedging of the LP position value instead of the impermanent loss. They study the problem of market-neutral liquidity provision by constructing a static replication portfolio that matches the AMM position's dollar value over time, effectively achieving a constant value position for the finite lifetime of the position. The replicating portfolio's margin requirement is then itself hedged by shorting a perpetual or dated futures contract to offset changes in the margin value, enabling a market-neutral LP strategy that generates interest from LP trading fees and futures funding fees. Meanwhile, Clark~\cite{clark2020replicating,clark2021replicating} investigates the replicating portfolio of the payoff of a constant product AMM position and shows how an LP position's terminal value can be fully statically replicated using a portfolio of European options. The approach focuses on fixed, finite time horizons and seeks to hedge only the terminal value of the liquidity position. On the other hand, our approach neither relies on dynamic hedging, nor assumes finite time horizons. Instead, we model the LP's position over an indefinite time horizon using a portfolio of perpetual American CI put options, providing a theoretical framework that statically captures the path-dependency of forward-looking LVR.

\section{Work Motivation}\label{sec:motivation}

The value profile of a concentrated liquidity position in a finite price band, as shown in Figure~\ref{fig:delta-gamma}, closely resembles that of a portfolio consisting of cash (a constant payoff) and short put options (the negative of the payoff shown in Figure~\ref{fig:ci-curve}): flat on one wing, linear on the other, and smoothly curved in between. Figure~\ref{fig:cpamm-vs-call} makes this visual similarity precise by comparing the \emph{value} of a CPAMM position with $k=1$ and band $(80,125)$ to a portfolio comprising cash and a one-month European put, with strike at the geometric mean of the price boundaries. The put is valued using the standard Black–Scholes model~\cite{black1973pricing}.

Despite the superficial similarity, key differences emerge. As shown in Figure~\ref{fig:diff-call-cpamm}, the two profiles diverge meaningfully. More critically, the European option’s value is inherently time-variant: even a single day’s passage erodes its time value (\emph{theta}), while the CPAMM’s value remains time-stationary. This contrast is illustrated in Figure~\ref{fig:theta}, where the solid line represents the option value as a function of the time-to-maturity. In addition, the CPAMM offers a flexible, perpetual holding period, whereas fixed-term options require periodic rolling---selling expiring contracts and buying new ones---to replicate a liquidity position.

These challenges raise a natural question: \emph{Can one construct a static portfolio that tracks a perpetual AMM band \emph{without} daily rebalancing?} The answer is affirmative, provided we replace European options with a class of \emph{American continuous-installment options}. This is because the ongoing constant funding fee, $q\;\mathrm{d}t$, charged by an active CI option offsets the time decay found in European options. In the limit of infinite maturity (the perpetual CI variant), the mark-to-market value of a CI put or call becomes \emph{time-invariant}, producing a flat line as shown by the dashed line in Figure~\ref{fig:theta}.

This observation enables a decomposition of a CFAMM band into a \emph{perpetual strip} of CI puts that both matches the pool’s delta and offers stationarity. Beyond its conceptual appeal, this decomposition yields two practical insights:
\begin{enumerate}
  \item The instantaneous \emph{funding fee} of the active CI option equals the LP’s \emph{loss-versus-rebalancing} cost. Since a CI option is perfectly hedgeable with a Black-Scholes type rebalancing porfolio~\cite{ciurlia2009note} and both the CPAMM and the CI option strip have the same delta, they share the same rebalancing portfolio, and hence the LVR is precisely the difference between the two instruments, \emph{i.e.} instantaneous funding fee.

  \item If the CFAMM’s delta and price boundaries are calibrated to match those of a single CI put, then the LVR over a future time window becomes nearly constant and equal to the funding fee $q$. Notably, this construction relies on implied volatility, rather than instantaneous volatility, meaning it can be formulated using observed market data. 
\end{enumerate}
The following section formalizes the CI decomposition, proving the equivalence between funding fees and LVR.
\begin{figure}[t]
  \centering
  \begin{subfigure}[b]{0.5\textwidth}
    \centering
    \includegraphics[width=\textwidth]{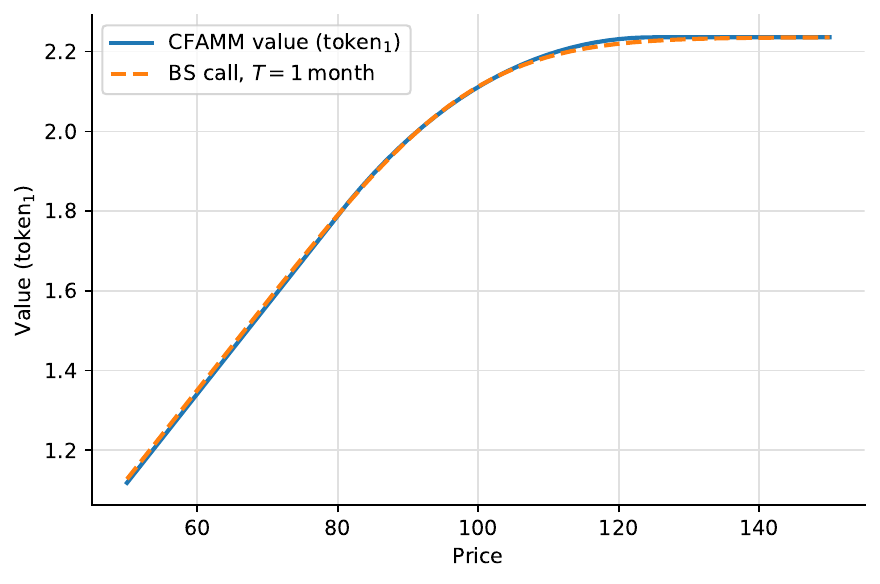}
    \caption{Value of a CPAMM band \([80,125]\) versus a short one-month Black–Scholes put plus cash.}
    \label{fig:cpamm-vs-call}
  \end{subfigure}
  \hfill
  \begin{subfigure}[b]{0.46\textwidth}
    \centering
    \includegraphics[width=\textwidth]{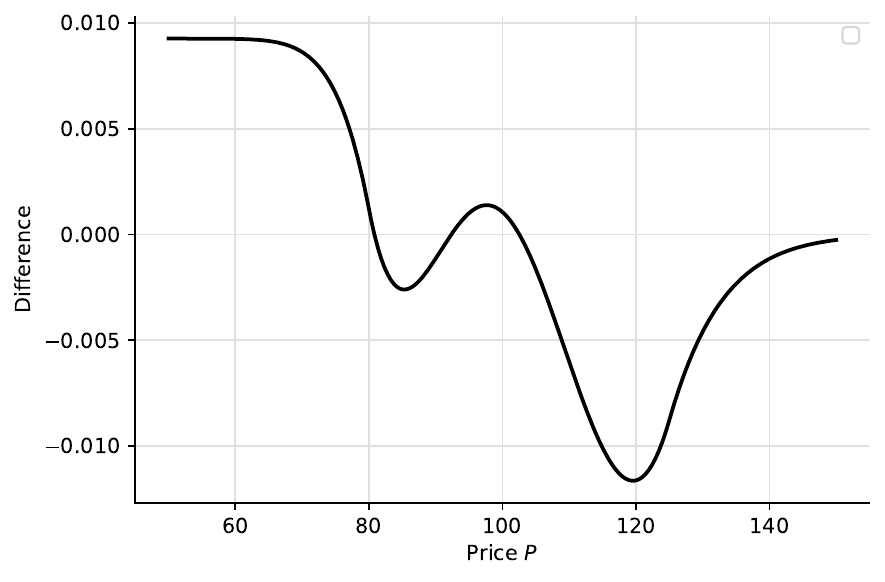}
    \caption{Difference between CPAMM value and short put plus cash.}
    \label{fig:diff-call-cpamm}
  \end{subfigure}
  \hfill
  \begin{subfigure}[b]{0.5\textwidth}
    \centering
    \includegraphics[width=\textwidth]{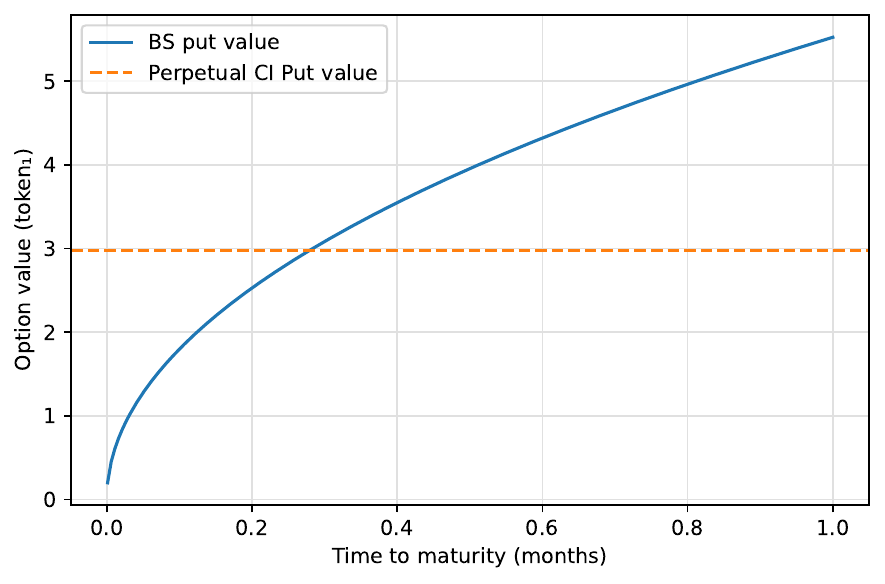}
  \caption{Time value decay of a European put versus the constant value of a perpetual CI call, whose funding fee offsets theta exactly.}
  \label{fig:theta}
  \end{subfigure}
  \caption{Comparison between CPAMM liquidity value and a portfolio of short put ($K{=}100$, $r{=}5\%$, $\sigma{=}50\%$) plus cash.}
  \label{fig:CPAMM-profile3}
\end{figure}

\section{Modeling a CFAMM Position with CI Options}\label{sec:kernel-replication}
\subsection{Overview}
In this section, we construct a portfolio whose delta (change in option price w.r.t. change in spot price) is the same as the delta of a CFAMM $X(S)$. This portfolio consists of a distribution of perpetual American CI put options in the limit $q\rightarrow\infty$ across a continuum of strike prices. As the funding rate tends to infinity, the exercise and dropping boundaries of each option collapse to the strike price, and the option's valuation converges to $\max(K-S, 0)$. Consequently, the option's delta converges to a step function. This property enables the construction of a portfolio that replicates the delta of an arbitrary (but smooth) CFAMM payoff function.
\subsection{Portfolio  Construction}
We begin with the following lemma:
\begin{lemma}\label{lem:boundry-collapse}
   In the limit $q\to\infty$, both the lower and upper boundary of a CI put converge to the strike $K$, and option's delta $X_{q}(S;K)$ becomes a step function $-\mathbf 1_{\{S<K\}}$.
\end{lemma}
This lemma encapsulates the relationship between the funding rate and the option’s “effective” time to expiration. As the funding rate increases, the option behaves increasingly like a zero-time-to-expiry option, with its delta approaching a discontinuous step as $q\to\infty$.

Next, we specify the weight distribution of the options portfolio that delta-replicates a CFAMM value function $V(S)$.
\begin{theorem}\label{thm:CFAMM-option-portfolio}
Let $V:\RR_{>0}\to\RR$ be a twice continuously differentiable function. Assume that $V'\in L^1(\RR_{>0})$, $V'$ has bounded variation on $\RR_{>0}$ and that $\displaystyle \lim_{S\to\infty} V'(S) = 0$. Define the weight $w(K):=V''(K)$ and, for each
$q<\infty$,
\begin{align*}
      \Pi_q(S)&:=\int_{0}^{\infty} w(K)\,P_q(S;K)\,dK,
\end{align*}
and
\begin{align*}
  \Pi(S)&:=\lim_{q\to\infty}\Pi_q(S).
\end{align*}

\noindent Then the portfolio
\(
   \Pi(S) -V(S)
\)
is delta-neutral.
\end{theorem}
Thus, a perpetual CFAMM position that can be closed by its owner at any time can be perfectly modeled using a distribution of CI puts with very large funding rates\footnote{$\Pi_{q}(S_{t})$ has the same payoff as described above at all times. Therefore, the holder must continuously reissue options that are exercised or dropped.}. In practice, a continuous distribution is infeasible, and funding rates are finite. Therefore, one may replicate the delta profile using a discrete set of puts with different strikes. However, discretization and finite funding rates introduce non-negligible delta-replication error, dependent on inter-strike spacing and funding rate. We quantify this approximation error numerically for a constant product AMM position in Section~\ref{sec:error-analysis-delta}.

An additional corollary of the above result is the relationship between the instantaneous LVR of a liquidity position and the funding fees of options with strikes around the spot price (also referred to as activated strikes) in the replicating portfolio. This relationship is analyzed in the following section.


\section{Establishing LVR as Funding Fees}\label{sec:LVR-from-CI}
The portfolio $\Pi$ delta-replicates a given CFAMM liquidity position. However, unlike the AMM position, the short options portfolio pays a continuous funding fee to the seller---arising from the active CI puts with strike around the spot price. We show that these funding fees, absent in the AMM, are precisely equal to the LP’s LVR. 

To compute this running funding fee, we first establish the following lemma:

\begin{lemma}\label{lem:funding-lvr-0}
As $q\to\infty$, the product $q\left(S_{u}(q;K) - S_{\ell}(q;K)\right)$ converges to a finite limit:
\[
\lim_{q \to \infty} q \cdot \bigl(S_{u}(q;K) - S_{\ell}(q;K)\bigr)
= \frac{\sigma^{2} K^{2}}{2}.
\]
\end{lemma}
Next, we express $\Pi$ as the limit of a discrete sum of options. In the following lemma, we construct such a discretization and show that, as $q \to \infty$, this portfolio converges pointwise to that of $\Pi$. We also compute the limiting funding fee contribution from the option whose holding region contains the spot price.

\begin{lemma}\label{lem:funding-lvr}
Fix $q>0$ and an interval $[a,b]\subset\RR_{>0}$. 
Define a sequence of strikes \( K_1 < K_2 < \dots < K_{N(q)} \) recursively by:
\[
  \begin{cases}
    S_\ell(q; K_1) = a, \\[4pt]
    S_\ell(q; K_{i+1}) = S_u(q; K_i), & i = 1, \dots, N(q)-1, \\[4pt]
    S_u(q; K_{N(q)}) = b.
  \end{cases}
\]
For each \( i \), define the weight
\[
   w_i := \int_{S_\ell(q; K_i)}^{S_u(q; K_i)} w(K)\,\mathrm{d}K
        = X\bigl(S_u(q; K_i)\bigr) - X\bigl(S_\ell(q; K_i)\bigr),
\]
so that
\[
   \sum_{i=1}^{N(q)} w_i = \int_a^b w(K)\,\mathrm{d}K = -1.
\]

Let \( K_j \) be the activated strike such that \( S_\ell(q; K_j) \le S_t \le S_u(q; K_j) \). Then, as \( q \to \infty \), the discrete portfolio 
\[
     \widetilde\Pi_{q}:=\sum_{i=1}^{N(q)}w_i\,P_q(\cdot;K_i)
\]
\noindent converges pointwise to the continuous payoff \( \Pi(S) \).
Moreover, the weighted funding fees at strike $K_j$ satisfies
\[
   \lim_{q \to \infty} w_j q = \frac{\sigma^{2} S^{2}}{2}X'(S).
\]
\end{lemma}

Lemma~\ref{lem:funding-lvr} constructs a portfolio of CI puts with discrete strikes and finite funding rates, with weights chosen such that, in the limit $q \to \infty$, the portfolio converges to the CFAMM-replicating portfolio defined in Theorem~\eqref{thm:CFAMM-option-portfolio}. The index $j$ denotes the unique option that remains active (\emph{i.e.,} in the holding region), while all other options are either exercised or dropped. Therefore, the funding fee received by the portfolio owner is $w_j q$.

As $q$ becomes large, Lemmas~\ref{lem:funding-lvr-0} and~\ref{lem:funding-lvr} together imply that this funding fee converges to the instantaneous LVR of the corresponding CFAMM position. This leads to the following result:

\begin{theorem}[Funding fee $=$ LVR]\label{th:funding-lvr}
Let $\mathrm{d}\LVR_t$ denote the instantaneous change in the LVR of a CFAMM position, and let $\mathrm{d}\Fee_t$ denote the instantaneous funding income of its delta-replicating CI option portfolio $\Pi$. Then,
\[
\mathrm{d}\Fee_t = \mathrm{d}\LVR_t
\quad\text{and}\quad
\Fee|_{0}^{T} = \LVR|_{0}^{T}
\qquad (\forall\, T > 0).
\]
\end{theorem}

Therefore, the LVR of a liquidity position is precisely the CI funding premium of its delta-replicating options portfolio. Another way to look at this is that a CI option can be perfectly delta-hedged using a continuously rebalanced Black--Scholes-type portfolio of risky and stable assets~\cite{ciurlia2009note}. As the funding rate \( q \) increases, the delta of this rebalancing portfolio converges to that of a CFAMM position. However, unlike the CI option, a CFAMM position does not compensate the liquidity provider via a funding stream. The discrepancy between the CFAMM position and its hedge is therefore exactly the foregone CI funding. This equality relies solely on closed-form expressions and the principle of self-financing, without invoking any additional modeling assumptions.

One advantage of this option-theoretic interpretation is that CI fee rates, as implied by the options market, provide both real-time and forward-looking estimates of LVR, enabling informed range management. Moreover, this framework permits the construction of liquidity profiles with nearly price-path-independent LVR. These features are demonstrated in the sections that follow. A further implication is that a market for CI options allows for static-weight delta-hedging of CFAMM positions---under the assumption of constant implied volatility---eliminating the need for frequent rebalancing, unlike with conventional American or European options.


\subsection{CFAMM Position Replicating a Unit CI Option}
\label{subsec:CI-equality}

Consider a concentrated–liquidity CFAMM band whose delta, at every price level, matches the delta of a \emph{single} perpetual American CI \emph{put} option. Specifically, choose the liquidity bounds
$a<b$ such that
\begin{align}\label{eq:delta-match}
    X(S)\;&=\;V'(S)\;\equiv\;X_q\bigl(S;K_*\bigr)
   \qquad (S\in[a,b]),\\
    a &= S_l(q,K_*),\\
    b &= S_u(q,K_*).
\end{align}

for some strike $K_*$ and finite fee rate $q$.

\begin{theorem}\label{th:lvr-q}
    The instantaneous rate of change of the LVR of the above CFAMM position is approximately equal to the funding fees of the unit put option, up to a bounded approximation error. 
    That is, there exists a residual function $\epsilon(t)$ with $|\epsilon(t)| \le rK_{*}$, such that
    \[
     d\LVR_t = q\,\mathrm dt + \epsilon(t) dt,
    \]
\end{theorem}
Thus, the AMM liquidity position with the above delta profile suffers an \emph{almost flat, price-path-independent, volatility‑independent} LVR. Note that because $X_q(S;K_*)$ depends on the volatility parameter $\sigma$, as shown in Eq~\eqref{eq:DeltaCIput}, the calibrated boundaries $a,b$---and thus the entire delta curve $X(S)$---remain implicitly volatility-dependent. The residual error term, $\epsilon(\cdot)$, is bounded in magnitude by a constant and its relative magnitude is reported and discussed in Section \ref{sec:sigma-calibration-put}. The above liquidity profile is useful for LPs who want to estimate forward LVR and compare it with the \emph{expected future trading fees}. Lastly, constructing such an AMM profile requires estimates of future volatility. This can be approximated using a term structure of implied volatility gathered from the fixed-term options market. Section~\ref{sec:sigma-calibration-put} discusses this in detail and estimates the approximation error arising from the calibration between fixed-term and perpetual options' volatilities.

\section{Error Analysis of Discrete CI-Option Replication}
\label{sec:error-analysis-delta}

In this section we quantify the approximation error that arises when the
continuous–strike decomposition of a concentrated CFAMM is replaced by a
\emph{discrete} strip of perpetual American CI put options with \emph{finite} $q$.

\subsection{Sources of error}
We isolate two drivers of error: \emph{(i)} The installment rate \(q\) (large but finite), and \emph{(ii)} the inter‑strike spacing \(\Delta K\) of the discrete strip. For a given pair \((q,\Delta K)\), we construct a strip of discrete options and measure the absolute difference between the target delta (of the CFAMM) and the strip's delta. This is done over the active price band $S\in[a,b]$ and its maximum and root‑mean‑square values are plotted.

\subsection{Experimental methodology}
 We chose a concentrated liquidity AMM as our target CFAMM. Thus, the analytical delta
    \(X(S)=L(1/\sqrt S-1/\sqrt{b})\) for
    uniform liquidity on a price band $[a,b]$. We choose $a=80$, and $b=125$, and $L$ chosen so that
    $X(a)=1$, $X(b)=0$. We discretize the replication weight such that on each strike interval \([K_i,K_{i+1}]\) we set  
        \begin{equation}
         w_i
            = \int_{K_i}^{K_{i+1}}\!
              X'(K)\,\mathrm dK
            \;=\;
              X(K_{i+1})
              -X(K_{i}),
        \end{equation}
        ensuring that the discrete weights sum to the continuous
        integral. For each
    strike, we compute the short CI put delta, clipped to $\{-1,0\}$
    outside its continuation band. Let $X_{\text{strip}}(S; q, \Delta K) = \sum_i w_i\,X_q(S;K_i)$ be the delta of the strip at $S$. Define the error
    \[
  \varepsilon(S_j; q, \Delta K)
    = \bigl|X(S_j)
             - X_{\text{strip}}(S_j; q, \Delta K)\bigr|
\]

evaluated on grid \(\{S_j\}_{j=1}^N\) where \(N\!=\!2000\). We consider two error metrics:
    \begin{itemize}
      \item \emph{Maximum absolute error:}
        $\max_j\,\varepsilon(S_j; q,\Delta K)$.
      \item \emph{Root-mean-square error (RMSE):}
        $\sqrt{N^{-1}\sum_j\varepsilon(S_j;q,\Delta K)^2}.$
    \end{itemize}
We evaluate a parameter sweep
    $(q,\Delta K)\in\{8 ,16, 32, 64, 125,250,500,1000, 2000,4000\}\times\{0.25,0.5,1.0, 2.0, 4.0\}$.

\subsection{Results}
Figure~\ref{fig:error-logmax} shows the logarithm of maximum absolute error versus
\(q\) for five strike spacings. Similarly, Figure~\ref{fig:error-rmse} displays the RMSE versus \(q\) for the same strike spacings. Both errors increase strictly with strike spacing for a given $q$ across all installment rates. On the other hand, for a given strike spacing, both errors generally decrease with $q$ with some exceptions. For large strike spacing, $\Delta K = \{1,2,4\}$, the RMSE error increases with $q$ for large values, $q\ge 128$. Lastly, Figure~\ref{fig:band-fit} plots a representative delta curve
(\(q{=}250,\ \Delta K{=}2\)) against the CPAMM target delta.

\begin{figure}[t]
  \centering
  \begin{subfigure}[b]{0.45\textwidth}
    \centering
    \includegraphics[width=\textwidth]{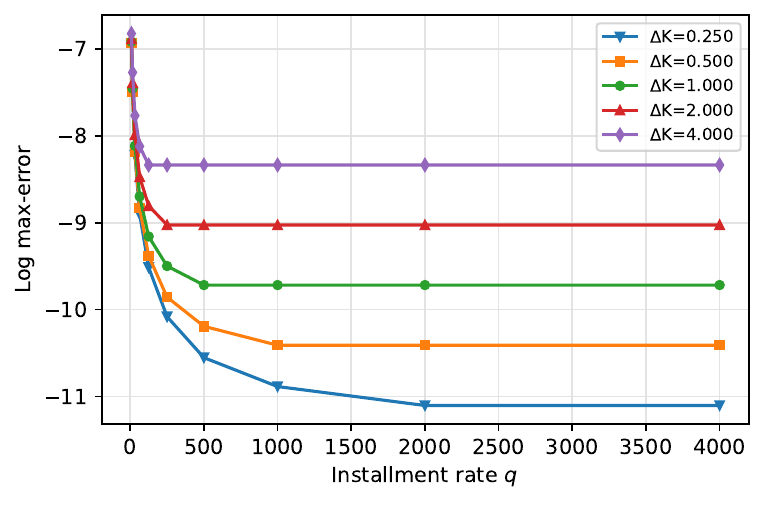}
    \caption{Log max error}
    \label{fig:error-logmax}
  \end{subfigure}
  \hfill
  \begin{subfigure}[b]{0.45\textwidth}
    \centering
    \includegraphics[width=\textwidth]{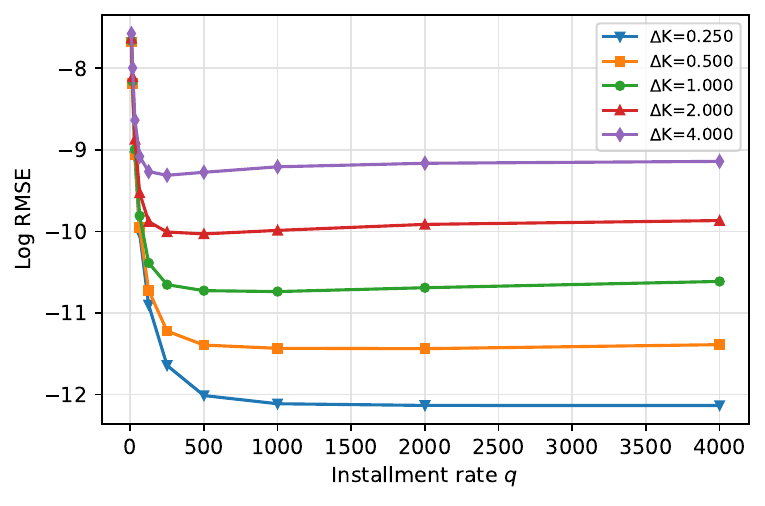}
    \caption{Log RMSE error}
    \label{fig:error-rmse}
  \end{subfigure}
    \hfill
  \begin{subfigure}[b]{0.55\textwidth}
    \centering
\includegraphics[width=\textwidth]{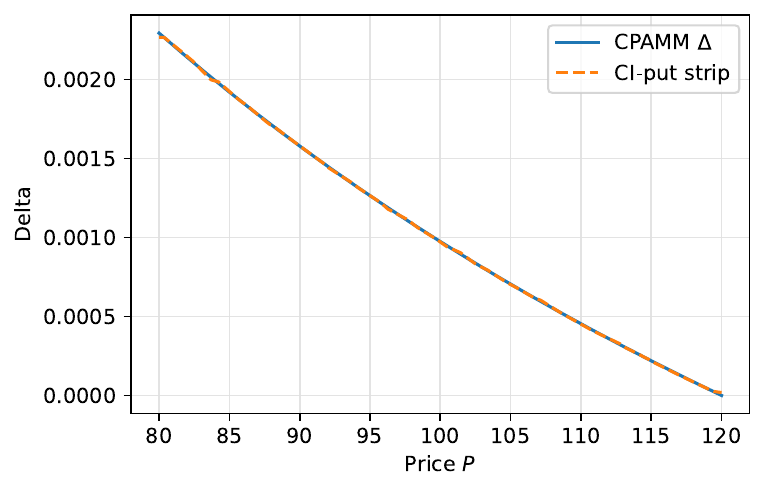}
  \caption{CFAMM target delta (solid) versus discrete CI‑put strip
           approximation (dashed) for \(q=250,\ \Delta K=2\).}
  \label{fig:band-fit}
  \end{subfigure}
  \caption{Log-Max absolute and RMSE delta-replication error versus installment rate $q$ under different parameter settings.}
  \label{fig:error-vs-q}
\end{figure}

\subsection{Discussion}
Both errors are below \(10^{-3}\) and shrink as expected: a larger \(q\) steepens each put’s delta,
while a finer \(\Delta K\) better resolves the continuous weight. The trade-off between the capital cost of a large $q$ and the operational cost of a finer strike mesh can be balanced according to the LP’s precision requirements.

\section{Volatility Calibration for Perpetual CI Put Options}
\label{sec:sigma-calibration-put}

Perpetual CI \emph{put} pricing requires a
\emph{constant} volatility parameter $\sigma$.  Market quotes instead
supply a term structure $\widehat\sigma(\tau)$ of implied volatility with time to expiration $\tau$. In the following, we derive the effective implied volatility $\sigma_{\mathrm{eff}}(q)$ for a perpetual American CI option with funding fee $q$ using the market-implied term structure of at-the-money options. 

\subsection{Effective time horizon}\label{sec:sigma–horizon-map}
For a CI put with rate~$q$, the continuation band is
$(S_{\ell}(q),S_{u}(q))$.

It stays alive as long
as the underlying spot price \(S_t\) remains within the continuation band. Otherwise, when $S_t = S_l$, the option is dropped, or when $S_t = S_u$, it is exercised by the holder. 

Define the \emph{first‑exit time}
\[
   \tau(q)\;:=\;
   \inf\!\bigl\{t>0 : S_t \notin (S_{\ell}(q),S_{u}(q))\bigr\},
\]
\emph{i.e.}\ the random horizon at which the CI position terminates.
In probabilistic terms, \(E[\tau(q)] = \bar \tau(q)\) is the \emph{mean first‑exit time} of a GBM between two absorbing boundaries. 

\begin{theorem}\label{th:tau-q-sol}
    The closed‑form solution of $\bar\tau(q)$ is
\begin{align}\label{eq:horizon-put}
    \bar\tau(q) = \begin{cases}
        \frac{1}{\sigma^2}\ln\left(\frac{S_0}{S_l(q)}\right)\ln\left(\frac{S_u(q)}{S_0}\right) & \text{if }a=0\\
        \frac{1}{a}\left[\ln \left(\frac{S_l(q)}{S_0}\right)+ \ln\left(\frac{S_l(q)}{S_u(q)}\right)\frac{S_0^{\kappa}-S_l(q)^{\kappa}}{S_l(q)^{\kappa}-S_u(q)^{\kappa}}\right] & \text{if }a \neq 0
    \end{cases} 
\end{align}
where $a = r - \frac{\sigma^2}{2}$ and $\kappa = -\frac{2a}{\sigma^2}$.
\end{theorem}

Figure~\ref{fig:exit_time} plots the distribution of the first-exit time, $\tau$ for a CI put with $r=2\%$, $\sigma=67\%$, $S_0 = K=100$, and $q=5$. Here, $\bar\tau = 1.6$ months, $\sqrt{\Var(\tau)} = 0.11$, and $\EE[|\tau - \bar\tau|] = 0.08$.

\begin{figure}[t]
  \centering
  \includegraphics[width=0.6\textwidth]{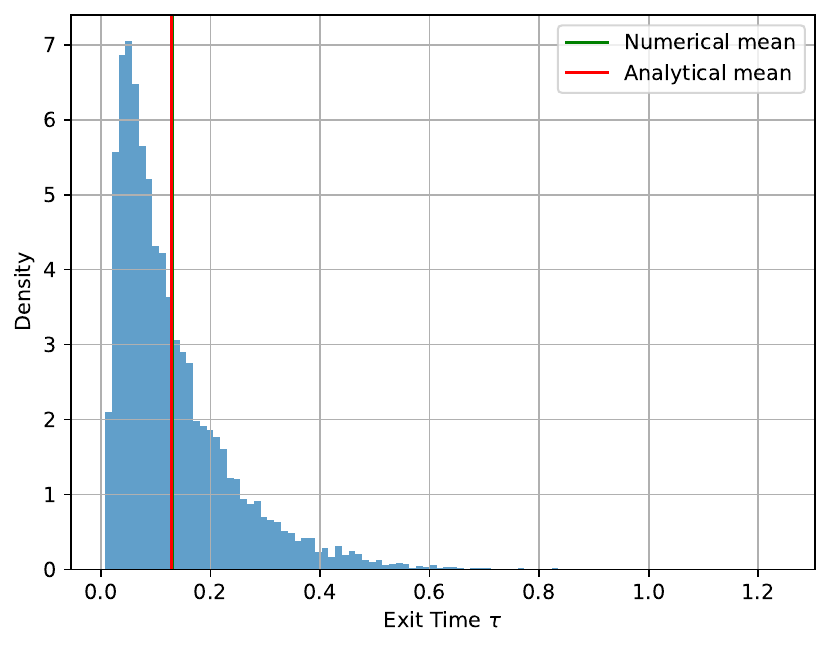}
  \caption{Distribution of first-exit time $\tau$ for $r=2\%$, $\sigma=67\%$, $K=100$, $q=5$.}
  \label{fig:exit_time}
\end{figure}
\begin{figure}[t]
  \centering
  \includegraphics[width=0.6\textwidth]{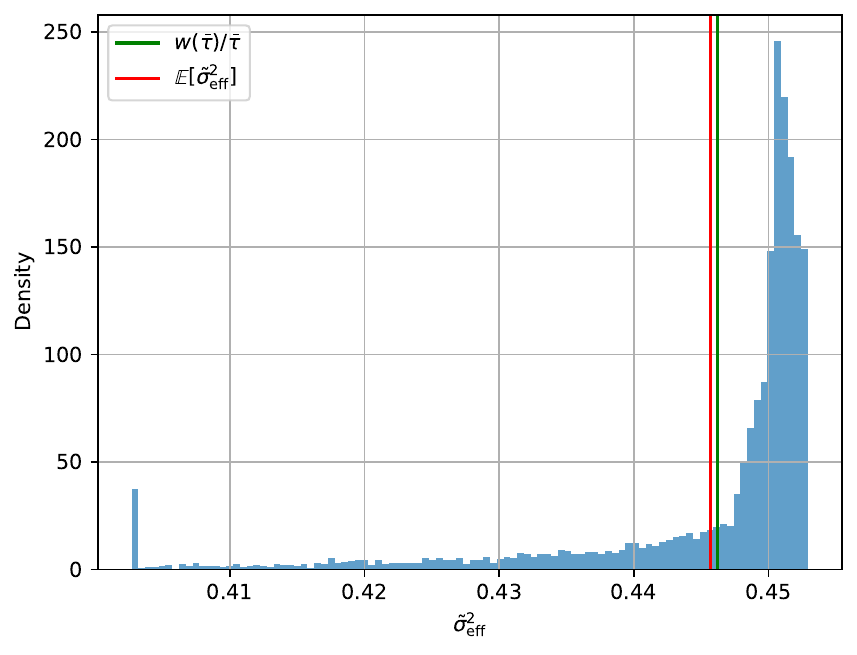}
  \caption{Distribution of $\sigma^2$ for $r=2\%$, $\sigma=61\%$, $K=100$, $q=5$ using ETH ATM IVs.}
  \label{fig:sigma_eff_eth}
\end{figure}

\subsection{Practical estimation from ATM IVs}
Given At-The-Money (ATM) implied volatilities $\hat \sigma(T)$ for fixed terms $T_1 < \ldots < T_n$ and a desired effective time horizon $\tau \in [T_i,T_{i+1})$, the squared constant volatility implied by the perpetual contract can be approximated by linearly interpolating the total variances ($T\hat\sigma^2(T)$) derived by the market implied volatilities:

\begin{align}
    \sigma_{\mathrm{eff}}^2 \tau \approx \tilde\sigma_{\mathrm{eff}}^2 \tau = \hat\sigma^2(T_i)T_i + \frac{\hat\sigma^2(T_{i+1})T_{i+1}-\hat\sigma^2(T_i)T_i}{T_{i+1}-T_i}(\tau - T_i) \equiv w(\tau)
\end{align}
The effective squared volatility can be estimated ex ante for a desired first-exit time distribution:
\begin{align}  
\label{eq:sigma_eff}
\EE[\tilde\sigma_{\mathrm{eff}}^2] = \EE\left[\frac{w(\tau)}{\tau}\right] \approx \frac{w(\bar\tau)}{\bar\tau}
\end{align}
When $\bar\tau$ is a function of $\sigma^2_{\mathrm{eff}}$, the problem becomes a fixed-point equation. Specifically, Eq.~\eqref{eq:sigma_eff} must be solved for $\sigma^2_{\mathrm{eff}}$ as a function of itself: $\sigma^2_{\mathrm{eff}} = \frac{w(\bar\tau(\sigma^2_{\mathrm{eff}}))}{\bar\tau(\sigma^2_{\mathrm{eff}})}$.

\begin{theorem}\label{th:error-bounds} 
    The estimate $\tilde\sigma_{\mathrm{eff}}^2\approx \frac{w(\bar\tau)}{\bar\tau}$ yields root mean squared error and mean absolute deviation
    \begin{align}
        RMSE & \leq M\sqrt{\Var(\tau)}\\
        MAD & \leq M \EE[|\tau-\bar\tau|]
    \end{align}
    where $M = \max_i \sup_{\tau \in [T_i,T_{i+1})} \left|\frac{d}{d\tau}\left(\frac{w(\tau)}{\tau}\right)\right|$.
\end{theorem}

Hence, when the total variance derived from market-implied volatilities (which are approximately linear in log-Moneyness $\log(K/S)$) has a small slope, the approximation error is small. Figure~\ref{fig:sigma_eff_eth} plots the distribution of $\tilde\sigma^2_{\mathrm{eff}}$, its mean, and the approximation $\frac{w(\bar\tau)}{\bar\tau}$ (for the same CI put as in Figure~\ref{fig:exit_time})
using ETH ATM 7-day, 30-day, 90-day, and 180-day IVs.
Figure~\ref{fig:sigma_eff_error} plots the $RMSE$ and $MAD$ (as a percentage of $\frac{w(\bar\tau)}{\bar\tau}$) for the approximation $\tilde\sigma^2_\mathrm{eff} \approx \frac{w(\bar\tau)}{\bar\tau}$ for the period of Jan 2024-Feb 2024.

\begin{figure}[t]
  \centering
  \includegraphics[width=0.6\textwidth]{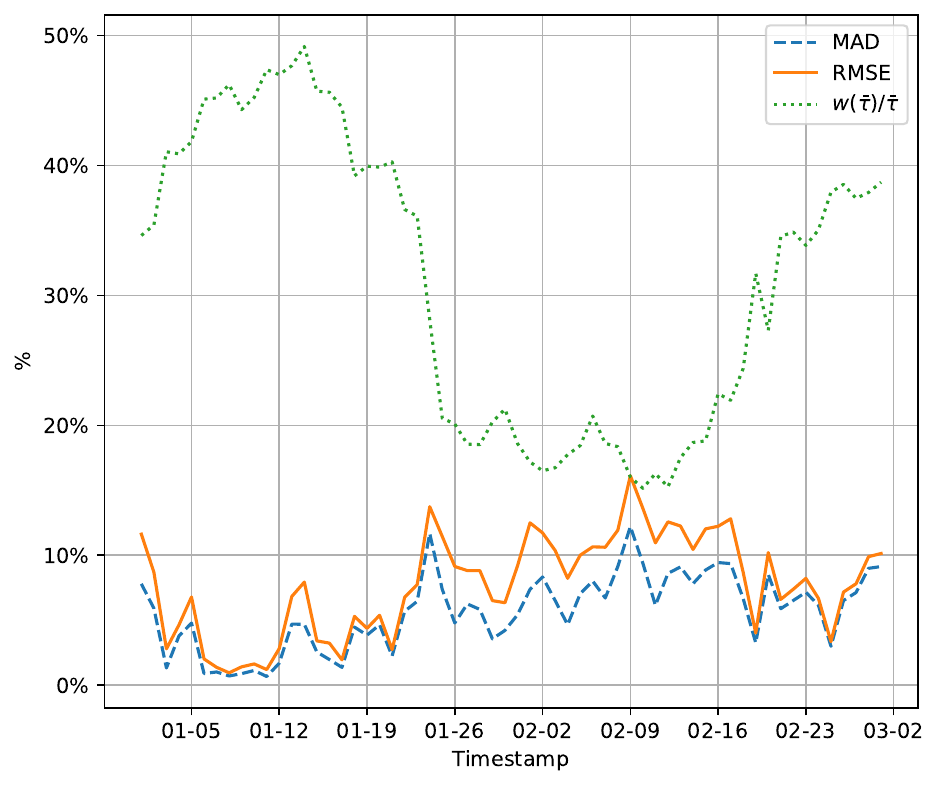}
  \caption{MAD and RMSE for $r=5\%$, $K=100$, $q=40$, where $\sigma_\mathrm{eff}^2$ is derived by fixed point methods. IV data is from ETH ATM IVs for Jan-Feb 2024.}
  \label{fig:sigma_eff_error}
\end{figure}

\begin{table}[t]
\centering
\caption{Funding fee $q$, resulting CFAMM band $(S_l,S_u)$, and residual bound $rK$, for $r=0.05$, for desired $\bar\tau$ under varying $\sigma_{\mathrm{eff}}$ exposures.}
\label{tab:ci-put-band}
\begin{tabular}{@{}ccccccc@{}}
\toprule
& & \multicolumn{4}{c}{\% of $K$} & \% of q \\ \cline{3-6}
$\bar\tau$ & $\sigma_{\mathrm{eff}}$ & $q$ (token$_1$/yr) &
$S_l(q)$ & $S_u(q)$ & Width & rK \\ \midrule

1 d & 60\% & 284\% & 97\% & 103\% & 6\% & 2\% \\
1 d & 80\% & 380\% & 96\% & 104\% & 8\% & 1\% \\
1 d & 100\% & 475\% & 95\% & 105\% & 10\% & 1\% \\
\midrule
1 wk & 60\% & 106\% & 92\% & 109\% & 17\% & 5\% \\
1 wk & 80\% & 142\% & 90\% & 112\% & 22\% & 4\% \\
1 wk & 100\% & 178\% & 87\% & 115\% & 28\% & 3\% \\
\midrule
2 wk & 60\% & 74\% & 89\% & 113\% & 24\% & 7\% \\
2 wk & 80\% & 99\% & 86\% & 118\% & 32\% & 5\% \\
2 wk & 100\% & 125\% & 83\% & 123\% & 40\% & 4\% \\
\midrule
1 mo & 60\% & 49\% & 85\% & 120\% & 35\% & 10\% \\
1 mo & 80\% & 66\% & 80\% & 128\% & 47\% & 8\% \\
1 mo & 100\% & 84\% & 76\% & 136\% & 60\% & 6\% \\
\midrule
2 mo & 60\% & 34\% & 79\% & 130\% & 51\% & 15\% \\
2 mo & 80\% & 46\% & 74\% & 142\% & 69\% & 11\% \\
2 mo & 100\% & 58\% & 68\% & 157\% & 88\% & 9\% \\
\bottomrule
\end{tabular}
\end{table}

\begin{figure}[h!]
  \centering
  \includegraphics[width=0.65\textwidth]{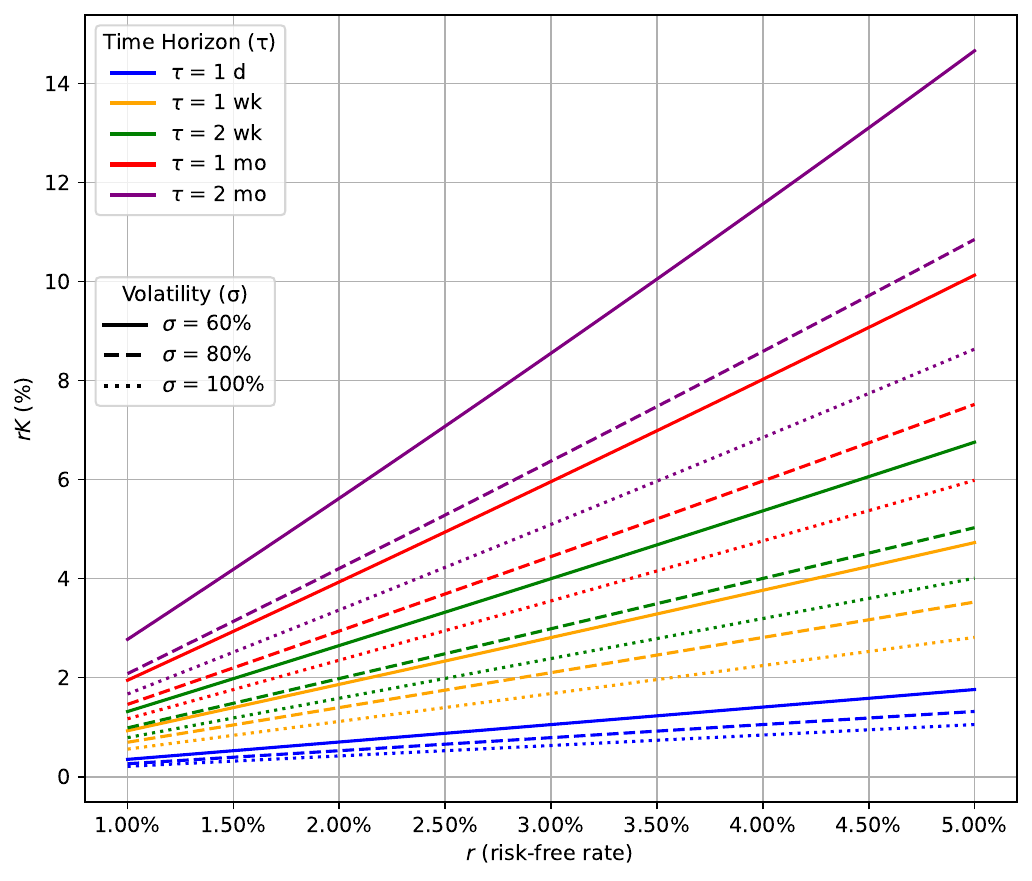}
  \caption{Residual upper bound $rK$ as a percentage of $q$ for various time horizons and $\sigma_\mathrm{eff}$ exposures.}
  \label{fig:rK-percent-q}
\end{figure}
\subsection{Interpretation for Liquidity Providers}

The volatility calibration framework enables liquidity providers (LPs) to estimate future loss-versus-rebalancing (LVR) of a concentrated AMM position using observable option market information. By associating the funding fee \( q \) of a perpetual CI put with its expected lifetime \( \bar\tau(q) \), and mapping this to market-implied volatilities, LPs can extract an estimate for the effective squared volatility \( \tilde\sigma^2_{\text{eff}} \) that governs the dynamics of the position and its underlying asset.

As the funding rate \( q \) increases, the continuation band \( [S_\ell(q), S_u(q)] \) narrows, leading to shorter expected lifetimes \( \tau(q) \) for the CI put. Conversely, smaller \( q \) implies wider bands and longer-lived options. Since market-implied volatilities are typically flatter at long durations, we can observe that:
\begin{itemize}
  \item \textbf{Short durations} (high \( q \)) correspond to low maturity implied volatilities, where the IV curve is more curved and error-prone. More fine-grained market data is required here for better estimates.
  \item \textbf{Long durations} (low \( q \)) correspond to long-dated IVs, where the volatility curve is typically flatter. The estimate $\tilde\sigma_{\mathrm{eff}}^2 = \frac{w(\tau)}{\tau}$ is then less sensitive to the exact value of $\tau$, so $M$ is small and the approximation of the mean is more robust to variation in $\tau$. Long durations also tend to have a tighter concentration of realised $\tilde\sigma^2_{\mathrm{eff}}$ around its mean for the same reason, meaning the mean is a good ex ante estimate.
\end{itemize}
If a liquidity band is chosen to replicate the delta of a single perpetual CI put---using the squared volatility estimate---then the LP incurs a predictable LVR almost equal to the funding rate \( q \). This transforms an otherwise stochastic adverse-selection cost into a predictable fixed cost per unit time, simplifying the LP’s decision-making. For instance, given a desired expected time-horizon, an LP can estimate the effective term volatility, which in turn informs the liquidity band selection. Table~\ref{tab:ci-put-band} shows the different band widths corresponding to expected time horizons and effective term volatilities and Figure~\ref{fig:rK-percent-q} shows the $\LVR-q$ residual term upper bound. Conversely, given a liquidity band, an LP can estimate the position's expected effective time horizon using numerical fixed point methods. Combined with the results of Section~\ref{sec:LVR-from-CI}, the effective term volatility estimation provides a practical framework for LPs to choose bands that realize a desired holding period \( \tau(q) \) and predictable LVR, or to estimate expected forward-looking LVR for a desired liquidity band.

\section{Conclusion}\label{sec:conclusion}

This work introduces a novel decomposition of a concentrated CFAMM position into a continuum of perpetual American continuous-installment (CI) put options, offering the first closed-form equivalence between the funding mechanics of perpetual American CI options and the loss-versus-rebalancing cost faced by concentrated AMM liquidity providers, providing as option-theoretic interpretation of LVR. 

By exploiting the limiting behavior of CI option valuations as the installment rate grows large, the paper constructs delta-replicating portfolios that match the AMM exposure exactly. The analysis shows that the funding income from this replicating portfolio, absent in the AMM, is analytically equal to the LVR cost borne by the LP in the AMM. This time-invariant correspondence permits a forward-looking estimation of LVR and provides actionable design rules for selecting position width and shape.

Beyond theoretical insight, the framework yields practical tools for LPs. The discrete error analysis confirms that a small collection of finite-$q$ CI puts suffices to replicate AMM delta within tight error bounds, making implementation of the replicating portfolio feasible for LPs wishing to immunize against LVR. Crucially, the analysis also shows that if a liquidity band is chosen to replicate the delta of a single perpetual CI put, the LP incurs a predictable LVR equal approximately to the funding rate $q$. This converts a stochastic adverse-selection cost into a predictable fixed cost per unit time, simplifying LP decision-making when it comes to position shape and width: using market-implied volatility curves, LPs can calibrate the position shape, width, and implied volatility to a desired expected holding period and LVR, or conversely, estimate a position’s effective time horizon and LVR given a liquidity band.

This framework opens several avenues for further research. While the present analysis focuses only on liquidity bands centered around the current price (ATM), LPs may, in practice, deploy liquidity asymmetrically about the spot price to expose their positions to greater or lesser volatility. A resulting mismatch of spot and the replicated strike requires extended analysis to capture the entire volatility surface, as opposed to only considering the ATM volatility curve. Alternatively, if on-chain markets for CI options were developed, they could serve as direct hedging instruments and sources of IV data. Furthermore, the Black-Scholes model assumes the underlying asset experiences a constant volatility, which is not supported by market data. In reality, volatilities may be time-dependent or even stochastic. Models like the Heston model or Hull-White model, which extend upon Black-Scholes, may be applied here for analysis under dynamic volatility surfaces and to assess the model's sensitivity to deviations from constant volatility assumptions. Finally, the model assumes continuous-time trading and perfect liquidity. Future work will relax these assumptions to quantify the impact of transaction costs, slippage, and gas fees on the CI funding fee-LVR equivalence.

\bibliographystyle{plain} 
\bibliography{refs} 

\appendix
\section{Closed-Form of Perpetual American CI Put Option}\label{app:closed-form}
Below, we provide closed-form expressions for the perpetual American CI put option.

\noindent Define

\begin{align}
  \gamma_p&:=-\frac{2r}{\sigma^{2}},\\ 
  g&:=1+\frac{rK}{q}.
\end{align}
Then, the option price is

\begin{align*}
  P_q(S;K)&=\alpha_pS+\beta_pS^{\gamma_p}+\frac{q}{r}
\end{align*}
where the constants $\alpha_p$ and $\beta_p$ are given by
\begin{align}
 \alpha_p&=\bigl(g^{1-1/\gamma_p}-1\bigr)^{-1}, \\
  \beta_p&=-\frac{1}{\gamma_p}\Bigl(\tfrac{q}{r+\sigma^{2}/2}\Bigr)^{1-\gamma_p}\alpha_p^{\gamma_p}.
\end{align}
Let $X_q(S;K)= \frac{\partial}{\partial S} P_q(S;K)$ be the delta of put value, its expression is given by
\begin{align*}
  X_q(S;K)=\alpha_p+\beta_p\gamma_pS^{\gamma_p-1}.
\end{align*}
Lastly, the lower and upper boundaries, $S_{\ell}$ and $S_u$ respectively, have the following closed-form:
\begin{align*}
  S_\ell &= \frac{q}{r+\sigma^{2}/2} \bigl[g - g^{1/\gamma_p}\bigr], \label{eq:boundary_lower} \\
  S_u    &= \frac{q}{r+\sigma^{2}/2} \bigl[g^{1-1/\gamma_p} - 1\bigr].      
\end{align*}

\section{Proofs of Main Theorems}
\begin{proof}[Proof of Lemma~\ref{lem:boundry-collapse}]
\textbf{Boundary collapse.}
Let \(S_\ell(q)\) and \(S_u(q)\) be the exercise and abandonment
boundaries in~\eqref{eq:boundary_lower} and~\eqref{eq:boundary_upper}. We will prove that
\[
   \lim_{q\to\infty}S_\ell(q)=
   \lim_{q\to\infty}S_u(q)=K.
\]
Define \(\varepsilon := rK/q\), so \(g=1+\varepsilon\) and
\(\varepsilon\downarrow 0\) as \(q\uparrow\infty\).
Consider the second‑order expansion
\((1+\varepsilon)^{a}=1+a\varepsilon+\tfrac12a(a-1)\varepsilon^{2}
        +\mathcal O(\varepsilon^{3})\).
Applying it to the two exponents in~\eqref{eq:boundary_lower} and~\eqref{eq:boundary_upper} yields
\begin{align*}
     S_\ell(q)&=
    \frac{q\varepsilon}{\,r+\sigma^{2}/2\,}
    \Bigl(1-\tfrac1{\gamma_p}\Bigr)
    +\mathcal O\!\bigl(q^{-1}\bigr),\\
  S_u(q)&=
    \frac{q\varepsilon}{\,r+\sigma^{2}/2\,}
    \Bigl(1-\tfrac1{\gamma_p}\Bigr)
    +\mathcal O\!\bigl(q^{-1}\bigr).
\end{align*}

\noindent Because \(q\varepsilon=rK\) and
\(1-\tfrac1{\gamma_p}=1+\sigma^{2}/(2r)\),
both leading terms equal \(K\). Therefore,
\begin{align*}
    \lim_{q\to\infty}S_\ell(q)&=\lim_{q\to\infty} K+\mathcal O\!\bigl(q^{-1}\bigr)\\ &= K,\\
    \lim_{q\to\infty}S_u(q)&=\lim_{q\to\infty} K+\mathcal O\!\bigl(q^{-1}\bigr) \\&= K.
\end{align*}

\noindent\textbf{Step-delta limit.} The smooth--fit conditions (continuous first order derivative) on the boundaries of the holding region of the option valuation curve from Section~\ref{sec:ode-options} give $X_{q}(S_\ell;K)=-1$ and $X_{q}(S_u;K)=0$.  Because $X_{q}$ is monotone increasing in $S$ between the two boundaries and the interval $S_u(q)-S_\ell(q)$ collapses, we have the \emph{pointwise} limit
\begin{equation}\label{eq:stepdelta}
X_{\infty}(S;K):=\lim_{q\to\infty}X_{q}(S;K)\;=
-\mathbf 1_{\{S<K\}}.
\end{equation}
Thus, as the funding fees of a continuous-installment put becomes large, it transforms into a unit-step-delta contract.
\end{proof}

\begin{proof}[Proof of Theorem~\ref{thm:CFAMM-option-portfolio}]
For each $K$, the map $S\mapsto P_q(S; K)$ is continuously differentiable. Moreover, for fixed $q$, there exists a constant $c_q > 0$ such that
\(
  \bigl|\partial_S P_q(S;K)\bigr|
  \le c_q(1+K)^{-2}|w(K)|.
\)
Hence, the integrand is point-wise dominated by an
$L^{1}$–function of $K$. Leibniz’s rule~\cite{folland1999real} yields
\[
  \partial_S\Pi_q(S)
   =\int_{0}^{\infty} w(K)\,X_q(S;K)\,dK,
\]
where $X_q:=\partial_S P_q$ is the CI–put delta. For every $K$, we have
\begin{align*}
 \lim_{q\to \infty} X_q(S;K)&=X_\infty(S;K)\\
 &=-\mathbf 1_{\{S<K\}}   
\end{align*}
from Lemma~\ref{lem:boundry-collapse}.  
Let $\bs{q_n}$ be a sequence such that $q_n\to\infty$.
$|X_{q_n}(S;K)|\le1$ for all $n$, $S$ and $K$, and $X_{q_n}(S;K) \to X_\infty(S;K)$ pointwise in $K$, so dominated convergence theorem~\cite{folland1999real} applies:
\begin{align*}
    \partial_S\Pi(S)
  &= \lim_{n\to\infty} \int_0^\infty w(K)X_{q_n}(S;K)dK \\
    &=\int_{0}^{\infty} w(K)\,X_\infty(S;K)\,dK\\
    &=-\int_{S}^{\infty} w(K)\,dK\\
    &=V'(S)-V'(\infty)\\
    &=V'(S).
\end{align*}
Therefore, \(\Pi(S)-V(S)\) is delta-neutral.
\end{proof}

\begin{proof}[Proof of Lemma~\ref{lem:funding-lvr-0}]
Using closed forms for a CI put boundaries from Appendix~\ref{app:closed-form}, $g:=1+\varepsilon$, $\varepsilon=rK/q$ and
$\gamma_p=-2r/\sigma^{2}$:
\begin{align*}
  S_{\ell}(q) &= \frac{q}{r+\sigma^{2}/2}\bigl[g-g^{1/\gamma_p}\bigr],\\
  S_{u}(q) &= \frac{q}{r+\sigma^{2}/2}\bigl[g^{1-1/\gamma_p}-1\bigr].
\end{align*}
Setting $\alpha:=1/\gamma_p=-\sigma^{2}/(2r)$ and expanding to second order:
\begin{align*}
 g^{\alpha} &= 1+\alpha\varepsilon+\tfrac12\alpha(\alpha-1)\varepsilon^{2}+\mathcal O(\varepsilon^{3}),\\
 g^{1-\alpha}      &= 1+(1-\alpha)\varepsilon-\tfrac12 (1-\alpha)\alpha\varepsilon^{2}+\mathcal O(\varepsilon^{3}).
\end{align*}
Substituting the above into the closed forms of $S_{\ell}$ and $S_{u}$ cancels the first‑order terms and the second‑order coefficient becomes
$\alpha(\alpha-1)$. Therefore, one obtains
\begin{equation}\label{eq:width-explicit}
  S_{u}(q)-S_{\ell}(q) =
   \frac{\alpha(\alpha-1)r^{2}K^{2}}{(r+\sigma^{2}/2)\,q}
   +\mathcal O\bigl(q^{-2}\bigr).
\end{equation}
Hence, 
\begin{align*}
    \lim_{q\to\infty} q\,\bigl(S_{u}(q)-S_{\ell}(q)\bigr) &=  \lim_{q\to\infty}\bigl[\frac{\alpha(\alpha-1)r^{2}K^{2}}{(r+\sigma^{2}/2)}
   +\mathcal O\bigl(q^{-1}\bigr)\bigr]\\
   &= \frac{\sigma^2K^2}{2}
\end{align*}
\end{proof}

\begin{proof}[Proof of Lemma~\ref{lem:funding-lvr}]
\textbf{Point‑wise convergence.}
Both $\widetilde\Pi_{q}$ and $\Pi$ \emph{vanish for all $S\ge b$}. Moreover, the CI–put delta satisfies
$0\le |X_q(S;K)|\le1$. Consequently
\(
  |\partial_{S}\widetilde\Pi_{q}(S)|
  =\bigl|\sum_i w_i\,X_q(S;K_i)\bigr|
  \le\sum_i |w_i| = 1
\)
for every $q$ and $S$. Let $i^\star=i^\star(S,q)$ be the (unique) index with
$S\in [S_\ell(q;K_{i^\star}),S_u(q;K_{i^\star})]$.
By Lemma~\eqref{lem:boundry-collapse}, $X_q(S;K_{i^\star})\to-\mathbf 1_{\{S<K_{i^\star}\}}$,
while $X_q(S;K_i)\to0$ for $i\neq i^\star$.  Dominated convergence
therefore gives
\[
   \lim_{q\to\infty}\partial_{S}\widetilde\Pi_{q}(S)
   =-\int_a^b w(K)\,\mathbf 1_{\{S<K\}}\,dK
   =X(S).
\]  
For any $S\le b$
\[
  \widetilde\Pi_{q}(S)=\widetilde\Pi_{q}(b)-\!\!\int_{S}^{b}\!
               \partial_{S}\widetilde\Pi_{q}(u)\,du.
\]
Because $\widetilde\Pi_{q}(b)=0$ for all $q$ and the integrands
converge point-wise while being uniformly bounded by 1, dominated
convergence implies $\widetilde\Pi_{q}(S)\to\Pi(S)$.

\noindent\textbf{Limit of $w_jq$.} From the statement of the Lemma, 
\[
w_j  \;=\;
            X\!\bigl(S_u(q;K_j)\bigr)
            -X\!\bigl(S_\ell(q;K_j)\bigr)
\]
Therefore,
\begin{align*}
    \lim_{q\to\infty} w_jq &=  \lim_{q\to\infty} q \left[X\!\bigl(S_u(q;K_j)\bigr)
            -X\!\bigl(S_\ell(q;K_j)\bigr)\right]\\
           &=  \lim_{q\to\infty} q(S_u(q;K_j) -S_\ell(q;K_j)) \frac{X\!\bigl(S_u(q;K_j)\bigr)
            -X\!\bigl(S_\ell(q;K_j)\bigr)}{(S_u(q;K_j) -S_\ell(q;K_j))}\\
            &= \lim_{q\to\infty} q(S_u(q;K_j) -S_\ell(q;K_j)) \lim_{q\to\infty}\frac{X\!\bigl(S_u(q;K_j)\bigr)
            -X\!\bigl(S_\ell(q;K_j)\bigr)}{(S_u(q;K_j) -S_\ell(q;K_j))}\\
            &= \lim_{q\to\infty} q(S_u(q;K_j) -S_\ell(q;K_j)) X'(K_{j})\\
            &= \frac{\sigma^2S_{t}^2}{2}X'(S_{t})
\end{align*}
\end{proof}

\begin{proof}[Proof of Theorem~\ref{th:funding-lvr}]
From Lemma~\ref{lem:funding-lvr}, we can approximate the continuous portfolio by a discrete portfolio of CI puts. For a specific funding fee $q$, at each time $t$, there is only one active option $j$ with weight $w_j$, price $P_q(S_t, K_j)$. The total funding fee accured over $[t,t+dt]$ is:
\[\dd\Fee_t^q = w_j q dt\]
Again from Lemma~\ref{lem:funding-lvr}, we have $\displaystyle \lim_{q\to\infty} w_jq = \frac{\sigma^2 S_t^2}{2}X'(S_t)$. This implies that
\[\lim_{q\to\infty} \mathrm{d}\Fee_t^q = \lim_{q\to\infty} w_j q dt = \frac{\sigma^2 S_t^2}{2}X'(S_t) dt\]
This exactly matches Eq.~\eqref{eq:dLVR}. Hence,
\[\dd\Fee_t = \dd\mathrm{LVR}_t\]
Integrating over $t\in [0,T]$, we get:
\[\Fee_t|_0^T = \int_0^T \dd\Fee_t = \int_0^T \dd\LVR_t = \LVR|_0^T\]
\end{proof}

\begin{proof}[Proof of Theorem~\ref{th:lvr-q}]
Equation~\eqref{eq:CI-ODE} gives
\begin{align*}
      \frac{1}{2}\,\sigma^{2}S^{2}\,
     \frac{\partial^{2}P_q}{\partial S^{2}}
  + rS\,
     \frac{\partial P_q}{\partial S}
  - r\,P_q
  \;=\; q,
  \qquad S\in(S_\ell,S_u).
\end{align*}
and
\begin{align*}
      \mathrm{d}\LVR_{t} &= \frac{1}{2}\,\sigma^{2}S^{2}\,
     \frac{\partial^{2}P_q}{\partial S^{2}} dt\\
     &= qdt-r(S\,
     \frac{\partial P_q}{\partial S} - P_q)dt
\end{align*}
Therefore, the residual term $\epsilon(t) = -r(S\,
     \frac{\partial P_q}{\partial S}- P_q)$. Note that we use $S$ instead of $S_t$ for brevity. $\epsilon(\cdot)$ is a function of $t$.

\noindent\textbf{Bounding $|\epsilon(t)|$.} In the region $S\in(S_\ell,S_u)$,
\[
\frac{\partial \epsilon(t)}{\partial S} = -rS\frac{\partial^2P_q}{\partial S^2}
\]
Since $\frac{\partial^2P_q}{\partial S^2}$ is always non-negative, 
\[
\frac{\partial \epsilon(t)}{\partial S} = -rS\frac{\partial^2P_q}{\partial S^2} \le0
\]
Therefore, $\epsilon(t)$ is monotonically decreasing with $S$. Moreover, when evaluated at $S=S_{\ell}$,
\begin{align*}
    \epsilon(t) &= -r((-1\cdot S_{\ell}) -(K_{*}-S_{\ell}))\\
    &= rK_{*}.
\end{align*}

and when $S=S_u$,
\[
\epsilon(t) = 0.
\]
Hence, $|\epsilon(t)|\le rK_{*}$.
\end{proof}

\section{Effective Time Horizon $\bar\tau(q)$}\label{sec:effective-horizon-proof}

\begin{proof}[Proof of Theorem~\ref{th:tau-q-sol}]
Given a continuation band,
\begin{equation}\label{eq:continuation-band}
   S_\ell(q) < S_t < S_u(q), \qquad  t \ge 0 ,
\end{equation}
and an initial spot price $S_0$, we formulate the problem as a GBM first-exit problem.

Under the risk‑neutral measure $\mathbb{Q}$ the spot follows a geometric
Brownian motion (GBM)
\begin{equation}\label{eq:gbm-riskneutral}
   dS_t = r S_t\,dt + \sigma S_t\,dB^{\mathbb{Q}}_t, \qquad
   S_0 \in\bigl(S_\ell(q),S_u(q)\bigr),
\end{equation}
where $\{B^{\mathbb{Q}}_t\}_{t\ge0}$ is a Wiener process.  Setting
$Y_t:=\log S_t$ converts \eqref{eq:gbm-riskneutral} into an \emph{arithmetic}
Brownian motion
\begin{equation}\label{eq:abm}
   dY_t = a\,dt + \sigma\,dB^{\mathbb{Q}}_t, \qquad
   a := r-\tfrac12\sigma^{2},
\end{equation}
absorbed at the logarithmic levels
\[
   L:=\log S_\ell(q), \qquad U:=\log S_u(q).
\]
The time at which the CI position terminates is the \emph{first‑exit time}
\begin{equation}\label{eq:def-tau}
   \tau(q)\;:=\;
   \inf\bigl\{t>0:Y_t\notin\bigl(L,U\bigr)\bigr\}
           \;=\;
   \inf\bigl\{t>0:S_t\notin\bigl(S_\ell(q),S_u(q)\bigr)\bigr\}.
\end{equation}

\subsection{Boundary Value ODE Problem}
Let $m(y) \equiv \EE_y[\tau(q)] \equiv \bar\tau(q)$ be the expectation of the time to absorption conditional on
$Y_0=y\in\bigl(L,U\bigr)$. By Dynkin's formula~\cite{ksendal2003stochastic}, 
\begin{align*}
    \EE_y[m(Y_\tau)] = m(y) + \EE_y\left[\int_0^\tau Am(Y_s) ds\right]
\end{align*}
\noindent At the exit time, $\tau$, the process has left $(L,U)$, so $m(Y_\tau) = 0$, and
\begin{align*}
    \EE_y[\tau] = m(y) = -\EE_y\left[\int_0^\tau Am(Y_s) ds\right]
\end{align*}
We can write $\tau$ as the integral of $1$ over $[0,\tau]$ so that
\begin{align*}
    \EE_y\left[\int_0^\tau 1 ds\right] = -\EE_y\left[\int_0^\tau Am(Y_s)ds\right] 
\end{align*}
which implies that $Am(y) = -1$. 

The generator $A$ of the time-homogeneous Brownian motion $Y_t$ also satisfies $Af(y) = a f'(y) + \frac{\sigma^2}{2}f''(y)$, yielding the boundary‑value problem:
\begin{equation}\label{eq:bvp}
   \frac{\sigma^{2}}{2}\,m''(y)+a\,m'(y)=-1,\qquad
   m\bigl(L\bigr)=m\bigl(U\bigr)=0 .
\end{equation}

\subsection{Closed‑form solution for critical drift $r=\tfrac12\sigma^{2}$}
When $r=\tfrac12\sigma^{2}$, $a=0$ in \eqref{eq:abm}, and the ODE becomes:
\begin{equation}\label{eq:critbvp}
    m''(y) = -\frac{2}{\sigma^2}
\end{equation}
The solution to \eqref{eq:critbvp} is of the form $m(y) = -\frac{1}{\sigma^2}y^2 + C_1y + C_2$. Applying the boundary conditions, $m(L) = m(U) = 0$, we have the system of equations:
\begin{align*}
    \begin{cases}
        \frac{1}{\sigma^2}L^2 = C_1L + C_2\\
    \frac{1}{\sigma^2}U^2 = C_1U + C_2
    \end{cases}
\end{align*}
Solving this yields
\begin{align*}
    C_1 &= \frac{1}{\sigma^2}(L+U)\\
    C_2 &= -\frac{LU}{\sigma^2}
\end{align*}
Finally, plugging in $L=\ln S_l(q)$, $U = \ln S_u(q)$, and $y = \ln S_0$, the solution to \eqref{eq:critbvp} is
\begin{equation}\label{eq:tau-final}
      \bar\tau(q)=\frac{1}{\sigma^2}\ln\left(\frac{S_0}{S_l(q)}\right)\ln\left(\frac{S_u(q)}{S_0}\right)
\end{equation}

\subsection{Closed-form solution for $r\neq\tfrac12\sigma^{2}$}
In this case, let $\kappa = -\frac{2a}{\sigma^2}$.
The general solution to the homogeneous ODE corresponding to \eqref{eq:bvp} is
\begin{equation*}
    m_h(y) = \frac{C_1}{\kappa}e^{\kappa y} + C_2
\end{equation*}
for some constants $C_1$ and $C_2$.

Let $m_p(y) = -\frac{y}{a} + C_3$, for some constant $C_3$, so that $m_p''(y) = 0$ and $m_p'(y) = -\frac{1}{a}$. $m_p$ is a particular solution to the ODE. Then, the general solution to the non-homogeneous ODE \eqref{eq:bvp} is
\begin{equation}
    m(y) = m_h(y)+m_p(y) = \frac{C_1}{\kappa}e^{\kappa y} + C_2 - \frac{y}{a} + C_3
\end{equation}
To simplify notation, let $A:=\frac{C_1}{\kappa}$ and $B:=C_2 + C_3$.
Applying the boundary conditions, $m(L) = m(U) = 0$, we have the system of equations:
\begin{align}
\begin{cases}
    Ae^{\kappa L} + B = \frac{L}{a}\\
    Ae^{\kappa U} + B = \frac{U}{a}
\end{cases}
\end{align}
Solving this yields
\begin{align}
    A & = \frac{1}{a} \frac{L-U}{e^{\kappa L}-e^{\kappa U}}\\
    B & = \frac{1}{a} \left(L - \frac{L-U}{e^{\kappa L}-e^{\kappa U}}e^{\kappa L}\right)
\end{align}
Finally, plugging in $L=\ln S_l(q)$, $U = \ln S_u(q)$, and $y = \ln S_0$, the solution to \eqref{eq:bvp} is
\begin{align}
    \bar\tau(q) = \frac{1}{a}\left[\ln \left(\frac{S_l(q)}{S_0}\right)+ \ln\left(\frac{S_l(q)}{S_u(q)}\right)\frac{S_0^{\kappa}-S_l(q)^{\kappa}}{S_l(q)^{\kappa}-S_u(q)^{\kappa}}\right]
\end{align}
\end{proof}

\section{Volatility Estimation Error Bounds}
\label{sec:error-bounds}
\begin{proof}[Proof of Theorem~\ref{th:error-bounds}]
The root mean squared error of the estimate $\tilde\sigma^2_{\mathrm{eff}} \approx \frac{w(\bar\tau)}{\bar\tau}$ is 
\begin{align*}
    RMSE = \sqrt{\EE\left[\left(\tilde\sigma_{\mathrm{eff}}^2 - \frac{w(\bar\tau)}{\bar\tau}\right)^2\right]}
\end{align*}
Let $f(\tau) = \tilde\sigma_{\mathrm{eff}}^2 = \frac{w(\tau)}{\tau}$, which is differentiable almost everywhere, with 
\begin{align*}
    f'(\tau) = \frac{m_i\tau - w(\tau)}{\tau^2} \qquad m_i \equiv \frac{\hat\sigma^2(T_{i+1})T_{i+1}-\hat\sigma^2(T_i)T_i}{T_{i+1}-T_i} \qquad f(0) = m_0
\end{align*}
$f$ is Lipschitz continuous, so let $M = \max_i \sup_{\tau \in [T_i,T_{i+1})} |f'(\tau)|$. Then, $|f(\tau) - f(\bar\tau)| \leq M|\tau - \bar\tau|$ and
\begin{align}
    RMSE = \sqrt{\EE[(f(\tau) - f(\bar\tau))^2]} \leq \sqrt{\EE[M^2(\tau - \bar\tau)^2]} = \sqrt{M^2\Var(\tau)} = M\sqrt{\Var(\tau)}
\end{align}

Similarly, the mean absolute deviation is
\begin{align}
    MAD = \EE\left[\left|f(\tau) - f(\tilde\tau)\right|\right] \leq M \EE[|\tau-\bar\tau|]
\end{align}
\end{proof}

\end{document}